\documentclass[10pt, journal, twocolumn]{IEEEtran}

\usepackage{longtable}
\usepackage{rotating}
\usepackage{multirow}
\usepackage{epsfig, amsmath,amssymb,epsf,cite,subfigure,scalefnt,multirow,array,setspace}
\usepackage{algorithm}
\usepackage{algorithmic}
\usepackage{times}

\usepackage{subfig}
\usepackage{amsfonts}
\usepackage{booktabs}
\usepackage{multirow}
\usepackage{multicol}
\usepackage{cite}
\usepackage{algorithm}
\usepackage{algorithmic}
\usepackage{stfloats}
\usepackage{bm}

\usepackage{color}
\usepackage{enumerate}

\usepackage{graphicx}
\usepackage{epstopdf}

\usepackage[final]{pdfpages}
\usepackage{fancyhdr}
\usepackage{lineno}

  \def\cC{{\mathcal{C}}}

 \def\cN{{\mathcal{N}}} \def\cO{{\mathcal{O}}}

\def\ba{{\mathbf{a}}}    \def\be{{\mathbf{e}}}
\def\bff{{\mathbf{f}}}    
   \def\bn{{\mathbf{n}}} 
 \def\bq{{\mathbf{q}}} \def\br{{\mathbf{r}}} \def\bs{{\mathbf{s}}} 
    \def\by{{\mathbf{y}}}
 \def\bh{\mathbf{h}}

\def\bA{{\mathbf{A}}} \def\bB{{\mathbf{B}}}   
\def\bF{{\mathbf{F}}}  \def\bH{{\mathbf{H}}} \def\bI{{\mathbf{I}}} 
  \def\bM{{\mathbf{M}}} \def\bN{{\mathbf{N}}} 
 \def\bQ{{\mathbf{Q}}} \def\bR{{\mathbf{R}}} \def\bS{{\mathbf{S}}} 
\def\bU{{\mathbf{U}}} \def\bV{{\mathbf{V}}} \def\bW{{\mathbf{W}}} \def\bX{{\mathbf{X}}} \def\bY{{\mathbf{Y}}}

\def\argmin{\mathop{\mathrm{argmin}}}
\def\argmax{\mathop{\mathrm{argmax}}}

\def\tr{\mathop{\mathrm{tr}}}

\def\rank{\mathop{\mathrm{rank}}}
\def\span{\mathop{\mathrm{span}}}
\def\vec{\mathop{\mathrm{vec}}}

\def\diag{\mathop{\mathrm{diag}}}

\renewcommand{\baselinestretch}{1.0}

     \def\d4{\!\!\!\!}

\def\bSig{\mathbf{\Sigma}}

  \def\R{{\mathbb{R}}} \def\C{{\mathbb{C}}}   \def\E{{\mathbb{E}}}

  \def\la{\left|}     \def\ra{\right|}    \def\lA{\left\|}     \def\rA{\right\|}

  \def\sig{\sigma}

  \def\-{\! - \!}  \def\+{\! + \!}  \def\={\! = \!}  \def\>{\! > \!}

\def \log{\mathrm{log}}

\newtheorem{proposition}{Proposition}
\newtheorem{theorem}{Theorem}
\newtheorem{lemma}{Lemma}

\newtheorem{remark}{Remark}

\newcommand{\bef}{\begin{figure}}
\newcommand{\eef}{\end{figure}}
\newcommand{\beq}{\begin{eqnarray}}
\newcommand{\eeq}{\end{eqnarray}}

\newenvironment{proof}[1][Proof]{\begin{trivlist}
\item[\hskip \labelsep {\bfseries #1}]}{\end{trivlist}}

\newcommand{\qed}{\nobreak \ifvmode \relax \else
\ifdim\lastskip<1.5em \hskip-\lastskip \hskip1.5em plus0em
minus0.5em \fi \nobreak \vrule height0.5em width0.5em
depth0.25em\fi}

\begin{document}
\begin{spacing}{1.0}
\pagenumbering{arabic}

\title{A Sequential Subspace Method for Millimeter Wave MIMO Channel Estimation}
\author{Wei Zhang, \IEEEmembership{Student Member,~IEEE}, Taejoon Kim, \IEEEmembership{Senior Member,~IEEE}, and Shu-Hung Leung
\thanks{

{W. Zhang is with the Department of Electrical Engineering, City University of Hong Kong, Hong Kong SAR, China (e-mail: wzhang237-c@my.cityu.edu.hk).}

{T. Kim is with the Department of Electrical Engineering and Computer Science, The University of Kansas, KS 66045, USA (e-mail: taejoonkim@ku.edu).}

{S.-H. Leung is with State Key
Laboratory of Terahertz and Millimeter Waves and Department of Electrical
Engineering, City University of Hong Kong, Hong Kong SAR, China
(e-mail: eeeugshl@cityu.edu.hk).}
}
}
\maketitle
\begin{abstract}
Data transmission over the millimeter wave (mmWave) in fifth-generation wireless networks aims to support
very high speed wireless communications.
A substantial increase
in spectrum efficiency for mmWave transmission can be achieved by using advanced hybrid analog-digital precoding, for which accurate channel state information (CSI) is the key. Rather than estimating the entire channel matrix, it is now well-understood that directly estimating subspace information, which contains fewer parameters, does have enough information to design transceivers. However, the large channel use overhead and associated computational complexity in the existing channel subspace estimation techniques are major obstacles to deploy
the subspace approach for channel estimation.
In this paper, we propose a sequential two-stage subspace estimation method that can resolve the overhead issues and provide accurate subspace information.
Utilizing a sequential method enables us to avoid manipulating the entire high-dimensional training signal, which greatly reduces the computational complexity.
Specifically, in the first stage, the proposed method samples the columns of channel matrix to estimate its column subspace. Then, based on the obtained column subspace, it optimizes the training signals to estimate the row subspace.
For a channel with $N_r$ receive antennas and $N_t$ transmit antennas, our analysis shows that the proposed technique only requires $\mathcal{O}(N_t)$ channel uses, while providing a guarantee of subspace estimation accuracy.
By theoretical analysis, it is shown that the similarity between the estimated subspace and the true subspace is linearly related to the signal-to-noise ratio (SNR), i.e., $\mathcal{O}(\text{SNR})$, at high SNR, while quadratically related to the SNR, i.e.,  $\mathcal{O}(\text{SNR}^2)$, at low SNR.
Simulation results show that the proposed sequential
subspace method can provide improved subspace accuracy, normalized mean squared error, and spectrum efficiency over existing
methods.

\end{abstract}


\maketitle

\section{Introduction}
\label{sec:introduction}
Wireless communications using the millimeter wave
(mmWave), which occupies the frequency band (30--300 GHz),
address the current scarcity of wireless broadband spectrum
and enable high speed transmission in fifth-generation (5G)
wireless networks \cite{rappaport}.
Due to the short wavelength, it is possible to employ large-scale antenna arrays with small-form-factor \cite{heath2016,Torkildson2011, Hur13}.
To reduce power consumption and hardware complexity, the mmWave systems exploit hybrid analog-digital multiple-input multiple-output (MIMO) architecture operating with a limited number of radio frequency (RF) chains \cite{heath2016}.
Under the perfect channel state information (CSI), {it has been shown that hybrid precoding can achieve nearly optimal performance as fully-digital precoding \cite{heath2016,Torkildson2011,spatially}.}
In practice, accurate CSI must be estimated via channel training in order to have effective precoding for robust mmWave MIMO transmission.
However, extracting accurate CSI in the mmWave MIMO poses new challenges  due to the limited number of RF chains that limits the observability of the channel and  greatly increases the channel use overhead.
\parskip=0pt

To reduce the channel use overhead, initial works focused on the beam alignment techniques \cite{Wang09,BeamSteer} utilizing beam search codebooks.
By exploiting the fact that mmWave propagation exhibits low-rank characteristic, recent researches formulated the channel estimation task as a {sparse signal reconstruction problem} \cite{OMPchannel,SBR_channel} {and} {low-rank matrix reconstruction problem} \cite{ZhangSD,AlternatingMin,recht,jointSparse,zhangSP, zhangSparse}.
By using the knowledge of sparse signal reconstruction, orthogonal matching pursuit (OMP) \cite{OMPchannel} and sparse
Bayesian learning (SBL) \cite{SBR_channel} were motivated to estimate the sparse mmWave channel in angular domain.
Alternatively, if the channel is rank-sparse, it is possible to directly
extract sufficient channel subspace information for the
precoder design \cite{hadi2015, ZhangSD,AlternatingMin}.
These subspace-based methods employ the Arnoldi iteration \cite{hadi2015} to estimate the channel subspaces and knowledge of matrix completion  \cite{ZhangSD, AlternatingMin} to estimate the low-rank mmWave channel information.

Though the sparse signal reconstruction \cite{OMPchannel,SBR_channel} and matrix completion \cite{ZhangSD,AlternatingMin} techniques can reduce the channel use overhead compared to traditional beam alignment techniques, the training sounders of these techniques \cite{OMPchannel,SBR_channel,ZhangSD,AlternatingMin} are pre-designed and high-dimensional, which leads to the fact that these works suffer from explosive computational complexity as the size of arrays grows.
To reduce the computational complexity, the adaptive training techniques have been investigated in \cite{Hur13,hadi2015,alk}, where the training sounders can be adaptively designed based on the feedback or two-way training.
But these adaptive training techniques could not guarantee the performance on mean squared error (MSE) and/or subspace estimation accuracy. {Moreover}, the techniques provided in \cite{Hur13,hadi2015,alk} will introduce additional channel use overhead due to the required feedback and two-way training.

To resolve the feedback overhead and maintain the benefit of adaptive training, in this paper, we present a two-stage subspace estimation approach, which sequentially estimates the column and row subspaces of the mmWave MIMO channel.
Compared to the existing channel estimation techniques in \cite{OMPchannel,SBR_channel,ZhangSD,AlternatingMin}, the training sounders of the proposed approach are adaptively designed to reduce the channel use overhead and computational complexity.
Moreover, the proposed approach is open-loop, thus it has no requirements of feedback and two-way channel sounding compared to priori adaptive training techniques \cite{Hur13,hadi2015,alk}.
The main contributions of this paper are described as follows:
\begin{itemize}
  \item We propose a two-stage subspace estimation technique called a sequential and adaptive subspace estimation (SASE) method.
      In the channel estimation of the proposed SASE,
the column and row subspaces are estimated sequentially.
Specifically, in the first stage, we sample a small fraction of columns of the channel matrix to obtain an estimate of the column subspace of the channel.
In the second stage,
the row subspace of the channel is estimated based on
the obtained column subspace.
In particular, by using the estimated column subspace obtained in the first stage, the receive training sounders of the second stage are optimized to reduce the number of channel uses.
Compared to the existing works with fixed training sounders, where the entire high-dimensional training signals are utilized to obtain the CSI, the proposed adaptation has the advantage that the dimension of signals being processed in each stage is much less than that of the entire training signal, greatly reducing the computational complexity.
Thus, the proposed SASE
has much less computational complexity than those of the existing methods.
  \item We analyze the subspace estimation accuracy, which guarantees the performance of the proposed SASE technique. Through extensive analysis,
       it is shown that the subspace estimation accuracy of the SASE is linearly related to
      the signal-to-noise ratio (SNR), i.e., $\mathcal{O}(\text{SNR})$, at high SNR, and quadratically related to the SNR, i.e., $\cO(\text{SNR}^2)$, at low SNR.
Moreover, simulation results show that the proposed SASE improves estimation accuracy over the prior arts.
\item
After obtaining the estimated column and row subspaces,
an efficient method is developed for estimating the high-dimensional
but low-rank channel matrix.
Specifically, given the subspaces estimated by the proposed SASE, the mmWave channel estimation task can be simplified to solving a low-dimensional least squares problem, whose computation is much lower. Simulation results show that the proposed channel estimation method has lower normalized mean squared error and higher spectrum efficiency than those of the existing methods.
\end{itemize}

This paper is organized as follows, in Section II, we introduce the mmWave MIMO system model. In Section III, the proposed SASE is developed and analyzed.
The channel use overhead, computational complexity, and an extension of the proposed SASE are discussed in Section IV.
 Finally, the simulation results and the conclusion remarks are provided in Sections V and VI, respectively.

\textit{Notation}: Bold small and captial letters denote vectors and
matrices, respectively.
$\bA^T,\bA^H, \bA^{\!-1}$, $| \bA  |$, $\| \bA \|_F$, $\tr(\bA)$,  and $\|\ba\|_2$ are, respectively,    the transpose, conjugate transpose, inverse, determinant, Frobenius norm, trace of $\bA$, and $l_2$-norm of $\ba$.
$[\bA]_{:,i}$, $[\bA]_{i,:}$, and $[\bA]_{i,j}$ are, respectively, the $i$th column, $i$th row, and $i$th row $j$th column entry of $\bA$.
$\vec(\bA)$ stacks the columns of $\bA$ and forms a column vector.
$\diag(\ba)$ denotes a square diagonal matrix with vector $\ba$ as the main diagonal.
$\sigma_L(\bA)$ denotes the $L$th largest singular value of matrix $\bA$.
$\bI_M \! \in \! \R^{M\times M}$ is the identity matrix.
{The $\mathbf{1}_{M,N} \! \in \! \R^{M\times N}$ , $\mathbf{0}_{M}\! \in \! \R^{M\times 1} , \mathbf{0}_{M,N} \! \in \! \R^{M\times N}$ are the all one matrix, zero vector, and zero matrix, respectively.}
$\mathop{\mathrm{col}}(\bA)$ denotes the column subspace spanned by the columns of matrix $\bA$. The operator $(\cdot)_+$ denotes $\max\{0,\cdot \}$. The operator $\otimes$ denotes the Kronecker product.

\section{MmWave MIMO System Model} \label{system section}

\begin{figure}[t]
\centering
\includegraphics[width=3.5in]{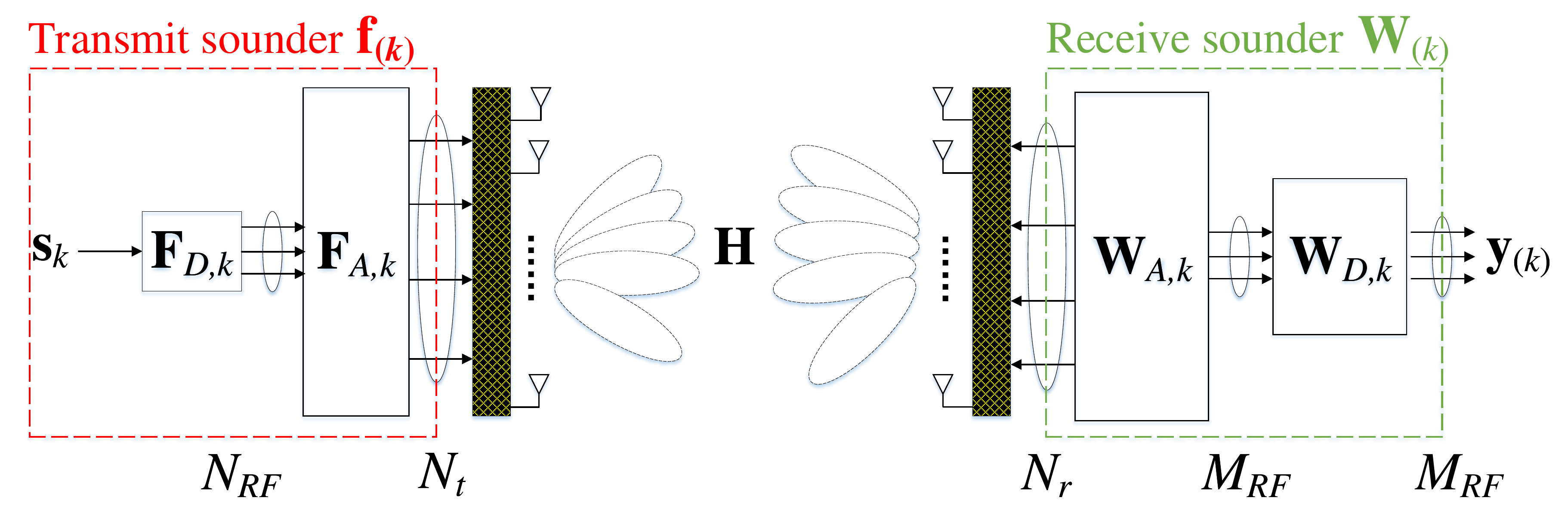}
\caption{The mmWave MIMO channel sounding model} \label{system}
\end{figure}

\subsection{Channel Sounding Model}

The mmWave MIMO channel sounding model is shown in Fig. \ref{system}, where the transmitter and receiver are equipped with $N_t$ and $N_r$ antennas, respectively. There are $N_{RF}\ge 2$ and $M_{RF} \ge 2$ RF chains at the transmitter and receiver, respectively.
Without loss of generality, we assume $N_t$ is an integer multiple of $N_{RF}$, and $N_r$ is also an integer multiple of $M_{RF}$.
In the considered mmWave channel sounding framework, one sounding symbol is transmitted over a unit time interval from the transmitter, which is defined as one channel use. It is assumed that the system employs $K$ channel uses for channel sounding.
The received signal $\by_{(k)} \in \C^{M_{RF} \times 1}$ at the $k$th channel use is given by
\beq
 \by_{(k)}= \bW_{(k)}^H\bH \bff_{(k)}+\bW_{(k)}^H \bn_{(k)},~k=1,\ldots,K, \label{training}
\eeq
where $ \bW_{(k)}\!\!= \!\! \bW_{A,k} \bW_{D,k} \in \C^{N_r \times M_{RF}}$ is the receive sounder composed of receive analog sounder $\bW_{A,k} \in \C^{N_r \times M_{RF}}$ and receive digital sounder $\bW_{D,k}\in \C^{M_{RF} \times M_{RF}}$ in series, $\bff_{(k)}\!\!= \!\! \bF_{A,k} \bF_{D,k} \bs_k \!\!\in\!\! \C^{N_t \times 1}$ is the transmit sounder composed of transmit analog sounder $\bF_{A,k} \in \C^{N_t \times N_{RF}}$ and transmit digital sounder $\bF_{D,k}\in \C^{N_{RF} \times N_{RF}}$ in series with transmitted sounding signal $\bs_k$, and $\bn_{(k)}\! \in \!\C^{{N_r} \times 1}$ is the noise.

Considering that the transmitted sounding signal $\bs_{k}$ is included in $\bff_{(k)}$, for convenience, we let $\bs_k \!\!= \!\!\frac{1}{\sqrt{N_{RF}}}[1, \ldots, 1]^T \in \R^{N_{RF}\times 1}$, which enables us to focus on the design of $\bF_{A,k}$ and $\bF_{D,k}$.
It is worth noting that the analog sounders are constrained to be constant modulus, that is, $|[\bW_{A,k}]_{i,j}|={1}/{\sqrt{N_r}}$, and $|[\bF_{A,k}]_{i,j}|={1}/{\sqrt{N_t}}, \forall i,j$. Without loss of generality, we assume the power of the transmit sounder is one, that is, $\| \bff_{(k)}\|_2^2 = 1$.
The noise $\bn_{(k)}$ is an independent zero mean complex Gaussian vector with covariance matrix  $\sig^2 \bI_{N_r}$.
Due to the unit  power of transmit sounder, we define the signal-to-noise-ratio (SNR) as $1/{\sigma^2}$.\footnote{Here, the SNR is the ratio of transmitted sounder's power to the noise's power, which is a common practice in the channel estimation literature \cite{Hur13,OMPchannel,SBR_channel,ZhangSD,hadi2015,alk}.} The details of designing the receive and transmit sounders for facilitating the channel estimation will be discussed in Section III.



{
To model the point-to-point sparse mmWave MIMO channel, we assume there are $L$ clusters with $L \ll \min\{N_r, N_t \}$, and each constitutes a propagation path. The channel model can be expressed as \cite{RobertOver, brady},
\beq
\bH = \sqrt{\frac{N_r N_t}{L}}\sum_{l=1}^{L} h_l \ba_r(\theta_{r,l}) \ba_t^H(\theta_{t,l}). \label{channel model}
\eeq
where $\ba_r(\theta_{r,l})\!  \in \!  \C^{N_r \! \times 1}\!$ and $\ba_t(\theta_{t,l}) \! \in  \!\C^{N_t \!     \times 1}\!$ are array response vectors of the uniform linear arrays (ULAs) at the receiver and transmitter, respectively. We extend it to the channel model with 2D uniform planar arrays (UPAs) in Section \ref{extension sec}. In particular, $\ba_r(\theta_{r,l})$ and $\ba_t(\theta_{t,l})$ are expressed as
\beq
\!\!\!\!&\ba_r(\theta_{r,l})\!=\!\frac{1}{\sqrt{N_r}}[ 1, e^{-j\frac{2\pi}{\lambda}d \sin \theta_{r,l}} , \!  \cdots \! ,e^{-j\frac{2\pi}{\lambda}d(N_r-1) \sin \theta_{r,l}}  ]^T,\nonumber\\
\!\!\!\!&\ba_t(\theta_{t,l})\!=\!\frac{1}{\sqrt{N_t}}[ 1, e^{-j\frac{2\pi}{\lambda}d \sin \theta_{t,l}} ,\! \cdots\! ,e^{-j\frac{2\pi}{\lambda}d(N_t-1) \sin \theta_{t,l}}  ]^T,\nonumber
\eeq
where $\lambda$ is the wavelength, $d  =  0.5\lambda$ is the antenna spacing, $\theta_{r,l}$ and $\theta_{t,l}$ are the angle of arrival (AoA) and angle of departure (AoD) of the $l$th path uniformly distributed in $[-\pi/2,\pi/2)$, respectively, and $h_l   \sim     \cC \cN(0,\sigma_{h,l}^2)$ is the complex gain of the $l$th path.

The channel model in \eqref{channel model} can be rewritten as
\beq
\bH=\bA_r \diag(\bh) \bA_t^H, \label{matrix expression}
\eeq
where $\!\!\bA_r\!\!=\!\![\ba_r(\theta_{r,1} ),\ldots,\ba_r(\theta_{r,L})] \!\!\in\!\! \C^{N_r \times L}$, $\bA_t\!\!=\!\![\ba_t(\theta_{t,1} ),\ldots,\ba_t(\theta_{t,L})]\!\in\!  \C^{N_t \times L}$}, and $\bh \!\!= \!\![h_1,\cdots,h_L]^T \!\!\in\!\! \C^{L \times 1}$. The channel estimation task is to obtain an estimate of $\bH$, i.e., $\widehat{\bH}$, from $\by_{(k)}$, $\bW_{(k)}$, and $\bff_{(k)}$, $k \!=\!\!1,\cdots,K$ in \eqref{training}.

\subsection{Performance Evaluation of Channel Estimation}
To evaluate the channel estimation performance, the achieved spectrum efficiency by utilizing the channel estimate $\widehat{\bH}$ is discussed in the following.
Conventionally, the precoder $\widehat{\bF}\! \in \! \C^{N_t \times N_d}$ and combiner $\widehat{\bW}\! \in \! \C^{N_r \times N_d}$ are designed, based on the estimated $\widehat{\bH}$, where $N_d$ is the number of transmitted data streams with $N_d\! \le \! \min\{ N_{RF},M_{RF} \}$.
Here, when evaluating the channel estimation performance, it is assumed the number of transmitted data streams is equal to the number of dominant paths, i.e., $N_d=L$.
After the  design of precoder and combiner, the received signal for the data transfer is given by
\beq
\by = \widehat{\bW}^H \bH \widehat{\bF} \bs + \widehat{\bW}^H\bn,  \label{receiver form}
\eeq
where the signal follows  $\bs\sim\cC\cN(\mathbf{0}_L, \frac{1}{L}\bI_{L})$ and $\bn \sim \cC \cN(\mathbf{0}_{N_r},\sigma^2 \bI_{N_r})$. It is worth noting that \eqref{receiver form} is for data transmission, while \eqref{training} is for channel sounding.
 The spectrum efficiency achieved by $\widehat{\bW}$ and $\widehat{\bF}$ in \eqref{receiver form} is defined in \cite{MIMO_capacity} as,
\beq
R = \log_2 \left|\bI_{L} + \frac{1}{\sig^2 L } \bR_n^{-1} \bH_e {\bH_e}^H \right| , \label{spectrum efficiency f}
\eeq
where $\bH_e \= \widehat{\bW}^H \bH \widehat{\bF} \in \C^{{L} \times {L}} $ and $\bR_n \=  \widehat{\bW}^H  \widehat{\bW}\in \C^{{L} \times {L}} $.
In this work, we assume that the precoder and combiner are unitary, such that $\widehat{\bW}^H\widehat{\bW}=\bI_{L}$ and $\widehat{\bF}^H\widehat{\bF}=\bI_{L}$. Under this assumption, we have $\bR_n  = \bI_{L}$ in \eqref{spectrum efficiency f}.

It is worth noting that the spectrum efficiency in \eqref{spectrum efficiency f} is invariant to the right rotations of the precoder and combiner, i.e., substituting $\widetilde{\bF} = \widehat{\bF} \bR_\bF$ and $\widetilde{\bW} = \widehat{\bW}\bR_\bW$ into \eqref{spectrum efficiency f}, where $\bR_\bF \in \C^{L \times L}$ and $\bR_\bW \in \C^{L \times L}$ are unitary matrices, does not change the spectrum efficiency. Thus, the $R$ in \eqref{spectrum efficiency f} is a function of subspaces spanned by the precoder and combiner, i.e., $\mathop{\mathrm{col}}(\widehat{\bF})$ and $\mathop{\mathrm{col}}(\widehat{\bW})$. Moreover, the highest spectrum efficiency can be achieved when $\mathop{\mathrm{col}}(\widehat{\bF})$ and $\mathop{\mathrm{col}}(\widehat{\bW})$ respectively equal to the row and column subspaces of $\bH$.

Apart from the spectrum efficiency achieved by the signal model in \eqref{receiver form}, we consider the effective SNR at the receiver,
\beq
\gamma = \frac{\| \widehat{\bW}^H \bH \widehat{\bF} \|_F^2}{\sigma^2 \| \widehat{\bW}\|_F^2}
= \frac{\| \widehat{\bW}^H \bH \widehat{\bF} \|_F^2}{\sigma^2 L}. \label{rec SNR}
\eeq
The received SNR $\gamma$ in \eqref{rec SNR} has the same rotation invariance property as the spectrum efficiency.
In other words, the $\gamma$  in \eqref{rec SNR} is a function of the estimated column and row subspaces.
The maximum of the $\gamma$ is also achieved when $\widehat{\bW}$ and $\widehat{\bF}$ span the column and row subspaces of $\bH$, respectively.

\begin{figure}[t]
\centering
\includegraphics[width=3.3in]{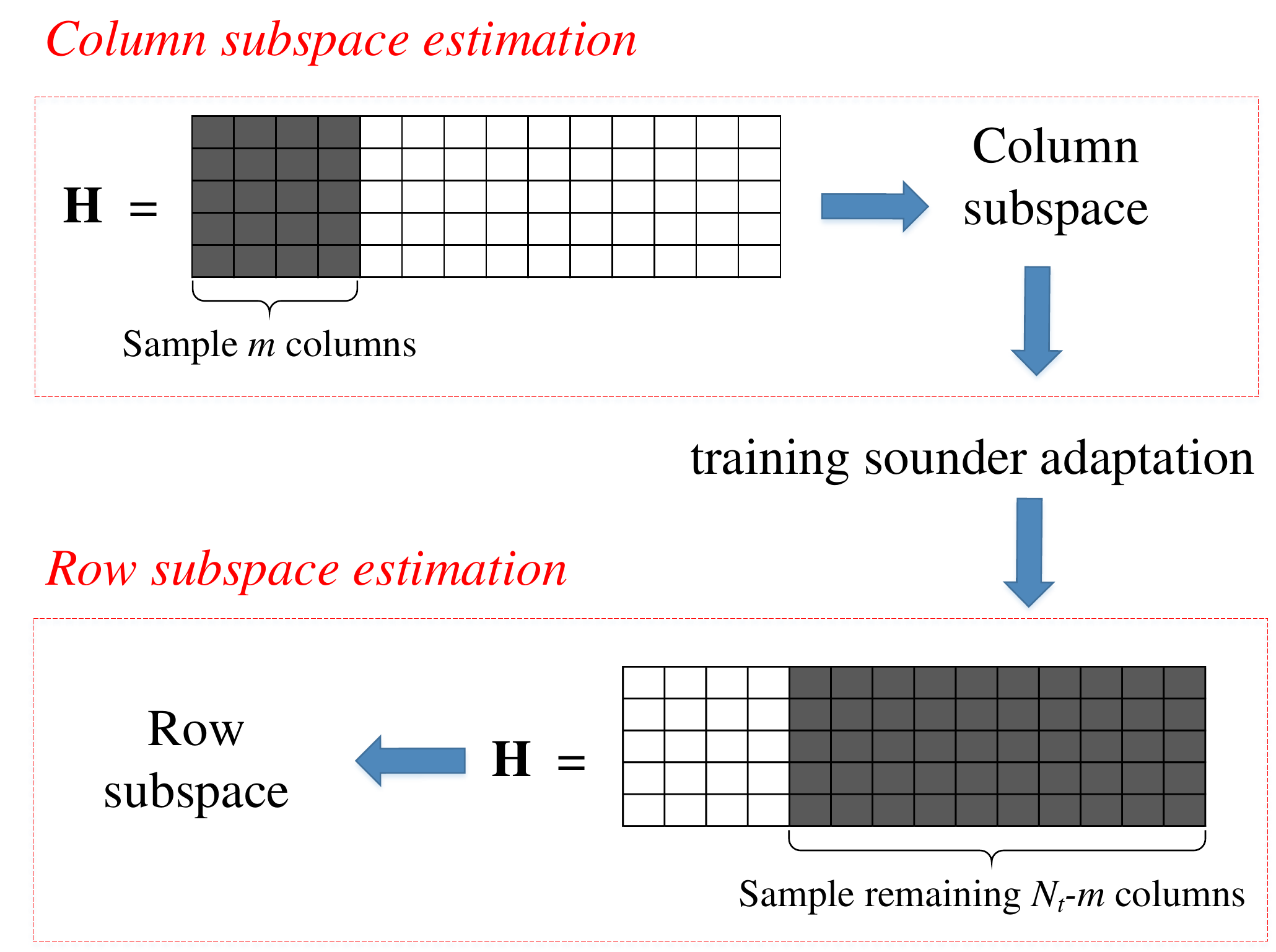}
\caption{Illustration of SASE Algorithm} \label{algorithm_diagram}
\end{figure}

Inspired by the definition in \eqref{rec SNR},
in this paper, the accuracy of subspace estimation is defined as the ratio of the power captured by the transceiver matrices \cite{mmvSubspace} $\widehat{\bW}$ and $\widehat{\bF}$ to the power of the channel,
\beq
\eta(\widehat{\bW} ,\widehat{\bF}) = \frac{\| \widehat{\bW}^H \bH \widehat{\bF} \|_F^2}{\tr(\bH^H \bH)}. \label{subspace metric}
\eeq
Similarly, the measures for the accuracy of column subspace and row subspace estimation, i.e., $\eta_c(\widehat{\bW})$ and $\eta_r(\widehat{\bF})$, are respectively defined as the ratio of the power captured by $\widehat{\bW}$ and $\widehat{\bF}$ to the power of the channel in the following,
\beq
\eta_c(\widehat{\bW}) &= &\frac{\tr( \widehat{\bW}^H \bH \bH^H \widehat{\bW})}{\tr(\bH^H \bH)}, \label{def eta c}\\
\eta_r(\widehat{\bF})& = &\frac{\tr( \widehat{\bF}^H \bH^H \bH \widehat{\bF})}{\tr(\bH^H \bH)}. \label{def eta r}
\eeq
Moreover, $\eta_c$ and $\eta_r$ are also rotation invariant.
When the values of $\eta_c$ or $\eta_r$ are closed to one, the corresponding  $\widehat{\bW}$ or $\widehat{\bF}$ can be treated accurate subspace estimates.

The illustration of the proposed SASE algorithm is shown in Fig. \ref{algorithm_diagram}.
It consists of two stages: one is column subspace estimation and the other is row subspace estimation.
In particular, the training sounders of the second stage can be optimized by fully adapting them to the estimated column subspace, which would reduce the number of channel uses and improve the estimation accuracy.


\section{Sequential and Adaptive Subspace Estimation}
\subsection{Estimate the Column Subspace} \label{Section_column}

In this subsection, we present the design of transmit and receive sounders along with the method for obtaining the column subspace of the mmWave channel.
{To begin with, the following lemma shows that under the mmWave channel model in (3), the column subspaces of $\bH$ and sub-matrix $\bH_S$ are equivalent.}

{
\begin{lemma}\label{lemma1}
Let $\bH_S = \bH\bS\in \C^{N_r \times m}$ be a sub-matrix that selects the first $m$ columns of $\bH$ with $m \ge L$, where $\bS $ is expressed as
\beq
\bS = \begin{bmatrix}
        \bI_{m} \\
        \mathbf{0}_{N_t-m , m}
      \end{bmatrix} \in \R^{N_t \times m}. \nonumber
\eeq
For the mmWave channel model in \eqref{matrix expression}, if
all the values of angles $\{\theta_{t,l} \}_{l=1}^L$ and $\{\theta_{r,l} \}_{l=1}^L$ are distinct, the column subspaces of $\bH$ and $\bH_S$ will be equivalent, i.e.,
$\mathop{\mathrm{col}}(\bH_S) = \mathop{\mathrm{col}}(\bH)$.
\end{lemma}
}
{
\begin{proof}
See Appendix \ref{appendix1}. \hfill \qed
\end{proof}
}

\begin{remark}
Because  $\{\theta_{t,l} \}_{l=1}^L$ and $\{\theta_{r,l} \}_{l=1}^L$ are continuous random variables (r.v.s) in $[-\pi/2,\pi/2)$, hence, they are distinct almost surely (i.e., with probability $1$).
\end{remark}

Lemma \ref{lemma1} reveals that when $\mathop{\mathrm{col}}(\bH_S) = \mathop{\mathrm{col}}(\bH)$, to obtain the column subspace of $\bH$,
it suffices to sample the first $m$ columns of $\bH$, i.e., $\bH_S$, which reduces the number of channel uses.
However, the mmWave hybrid MIMO architecture can not directly access the entries of $\bH$ due to the analog array constraints. This can be overcome by adopting the technique proposed in \cite{hadi2015}.
Specifically, to sample the $i$th column of $\bH$, i.e., $[\bH]_{:,i}$, the transmitter needs to construct the transmit sounder {$\bff_{(i)} =  \be_i   \in   \C^{N_t \times 1}$}, where $\be_i$ is the $i$th column of $\bI_{N_t}$.
{This is possible due to the fact that any precoder vector can be generated by $N_{RF} \ge 2$ RF chains \cite{xzhang}. To be more specific, there exists $\bF_{A,i}$, $\bF_{D,i}$, and $\bs_i$ such that $\be_{i}=\bF_{A,i} \bF_{D,i}\bs_i$,
\beq
\be_{i} \!\! = \!\! \underbrace{\frac{1}{\sqrt{N_t}} \!\!\begin{bmatrix}
            1 & 1 \!\!\!& \cdots\\
            \vdots & \vdots   \!\!\!& \cdots \\
            1 & 1   \!\!\! &\cdots\\
            1 & -1    \!\!\!& \cdots\\
            1 & 1    \!\!\!&\cdots\\
            \vdots & \vdots  \!\!\!& \cdots \\
            1 & 1 \!\!\! &\cdots
          \end{bmatrix}}_{\triangleq \bF_{A,i} }
          \underbrace{\begin{bmatrix}
            \frac{\sqrt{N_{RF} N_t}}{2}\!\!\!\!\!\! & 0 \!\!\!& \cdots \!\!\!& 0 \\
            -\frac{\sqrt{N_{RF} N_t}}{2}\!\!\!\!\!\!& 0 \!\!\!& \cdots \!\!\!& 0 \\
            0 \!\!\!\!\!\!& 0\!\!\! & \cdots \!\!\!& 0 \\
            \vdots\!\!\! \!\!\!& \vdots\!\!\! & \ddots \!\!\!& \vdots \\
            0\!\!\!\!\!\! & 0 \!\!\!& \cdots \!\!\!& 0
          \end{bmatrix} }_{\triangleq \bF_{D,i}}
             \underbrace{       \frac{1}{\sqrt{N_{RF}}}\begin{bmatrix}
            1  \\
            1  \\
            \vdots\\
           1
          \end{bmatrix}}_{\triangleq \bs_i}, \nonumber
\eeq
where $\bF_{A,i} \!\!= \!\!\frac{1}{\sqrt{N_t}}\mathbf{1}_{N_t, N_{RF}}$ except for $[\bF_{A,i}]_{i,2}\!\!=\!\!-\frac{1}{\sqrt{N_t}}$,
the $\bF_{D,i} \!\!=\!\! \mathbf{0}_{N_{RF}, N_{RF}}$ except for $[\bF_{D,i}]_{1,1}\!\!=\!\!\frac{\sqrt{N_{RF}N_t}}{2},[\bF_{D,i}]_{2,1}\!\!=\!\!-\frac{\sqrt{N_{RF}N_t}}{2}$, and $\bs_i \!\!=\!\! \frac{1}{\sqrt{N_{RF}}}[1, \ldots, 1]^T \in \R^{N_{RF}\times 1}$.}

At the receiver side, we collect the receive sounders of $N_r/M_{RF}$  channel uses to form the full-rank matrix,
\beq
\bM = [\bW_{(i,1)}, \bW_{(i,2)},\cdots,\bW_{(i,N_r/M_{RF})}] \in \C^{N_r\times N_r}, \label{Wij expression}
\eeq
where  $\bW_{(i,j)}\in \C^{N_r \times M_{RF}},~ j=1, \ldots ,N_r/M_{RF}$, denotes
the $j$th receive sounder corresponding to transmit sounder $\be_i$.
In order to satisfy the  analog constraint where the entries in analog sounders should be constant modulus, we let the matrix $\bM$ in \eqref{Wij expression} be the discrete Fourier transform (DFT) matrix.
Specifically, the analog and digital receive sounders associated with $\bW_{(i,j)}$ in \eqref{Wij expression} are expressed as follows
\beq
\bW_{(i,j)} =  \underbrace{[\bM]_{:, (j-1)M_{RF}+1:jM_{RF}}}_\text{analog sounder} \underbrace{\bI_{M_{RF}}}_\text{digital sounder}. \nonumber
\eeq

Thus, the received signal $ \by_{(i,j)} \!  \in  \!  \C^{M_{RF}   \times  1}$ under the transmit sounder $\be_i$ and receive sounder $\bW_{\!(i,j)}$ is expressed as follows
\beq
\by_{(i,j)} =\bW_{(i,j)}^H \bH \be_i + \bW_{(i,j)}^H\bn_{(i,j)} , \nonumber
\eeq
where $\bn_{(i,j) }\in \C^{N_r \times 1}$ is the noise vector with $\bn_{(i,j)} \sim \cC \cN(\boldsymbol{0}_{N_r},\sigma^2\bI_{N_r})$.
Then we stack the observations of  $N_r/M_{RF}$  channel uses as $\by_{i} = [\by_{(i,1)}^T, \cdots,\by^T_{(i,N_r/M_{RF})} ]^T \in \C^{N_r \times 1}$,
\beq
\underbrace{ \left[
\begin{matrix}
  \by_{(i,1)} \\
 \by_{(i,2)} \\
  \vdots \\
  \by_{(i,\frac{N_r}{M_{RF}})}
\end{matrix}
\right]}_{\triangleq \by_i}
\!\!\!\! &=& \!\!\!\!
\underbrace{\left[
\begin{matrix}
  \bW_{(i,1)}^H \\
 \bW_{(i,2)}^H \\
  \vdots \\
  \bW_{(i,\frac{N_r}{M_{RF}})}^H
\end{matrix}
\right]}_{\triangleq \bM^H}
\underbrace{\bH \be_i}_{\triangleq [\bH]_{:,i}} \!\!+\!\!
\underbrace{\left[
\begin{matrix}
  \bW_{(i,1)}^H \bn_{(i,1)}\\
 \bW_{(i,2)}^H \bn_{(i,2)}\\
  \vdots \\
  \bW_{(i,\frac{N_r}{M_{RF}})}^H\bn_{(i,\frac{N_r}{M_{RF}})}
\end{matrix}
\right]}_{\triangleq \tilde{\bn}_i}  \nonumber\\
\!\!\!\! &=& \!\!\!\!
\bM^H
[\bH]_{:,i} +
\tilde{\bn}_i \label{H1 observation},
\eeq
where $\tilde{\bn}_i \in \C^{N_r \times 1}$ is the effective noise vector after stacking, whose covariance matrix is expressed as,
\beq
\!\!\!\!\!\!\E[\tilde{\bn}_i \tilde{\bn}_i^H\!]\!\!\!\!\!\!  & =& \!\!\! \!\! \sig^2 \!\!
\!\begin{bmatrix} \!\!\!
 \! \bW_{(i,1)}^H\bW_{(i,1)}\!\!\!\! & \!\!\!\!\cdots\!\!\!\!&\!\! \bW_{(i,1)}^H\bW_{(i,\frac{N_r}{M_{RF}})}  \\
  \! \vdots \!\!\!\!& \!\!\!\!\ddots\!\!\!\! & \!\!\vdots  \\
\! \bW_{(i,\frac{N_r}{M_{RF}})}^H\bW_{(i,1)} \!\!\!\!& \!\!\!\!\cdots\!\! \!\!&\!\!\! \bW_{(i,\frac{N_r}{M_{RF}}\!\!)}^H\!\!\bW_{\!\!(i,\frac{N_r}{M_{RF}}\!\!)}
\!\!\!\end{bmatrix}\!\!.\label{cov effective n}
\eeq
Because the DFT matrix $\bM$ in \eqref{Wij expression} satisfies $\bM^H \bM = \bM \bM^H = \bI_{N_r}$, the following holds
\beq
\bW_{(i,j)}^H  \bW_{(i,k)} =
\left\{
\begin{array}{rcl}
\bI_{M_{RF}}       &      & {j      =     k},\\
\mathbf{0}_{M_{RF}}      &      & {j \neq k}.
\end{array} \right. \label{property cov}
\eeq
Substituting \eqref{property cov} into \eqref{cov effective n}, we can verify that $\E[\tilde{\bn}_i \tilde{\bn}_i^H] = \sigma^2 \bI_{N_r}$, and precisely, $\tilde{\bn}_i \sim \cC\cN(\mathbf{0}_{N_r},\sigma^2 \bI_{N_r})$. Moreover, by denoting $\widetilde{\bN} = [\tilde{\bn}_1,\cdots,\tilde{\bn}_m] \in \C^{N_r \times m}$, it is straightforward that the entries in $\widetilde{\bN}$ are independent, identically distributed (i.i.d.) as $\cC \cN(0, \sigma^2)$.
Here, for convenience, we denote $ \widetilde{\bY}_S  =  [\by_{1},\cdots, \by_{m}]  \in  \C^{N_r \times m}$ where $\by_i$ is defined in \eqref{H1 observation}.
Then, we apply DFT to the collected observation $\widetilde{\bY}_S$, and obtain $\bY_S = \bM\widetilde{\bY}_S  \in \C^{N_r \times m}$ as
\beq
\bY_S  =\bH_S + \bN_S, \label{colun observation}
\eeq
where $\bN_S \! =\! \bM \widetilde{\bN}  \in   \C^{N_r \times m}$ and $\bH_S\! =\! [\bH]_{:,1:m}  \in   \C^{N_r \times m}$.
Before talking about the noise part $\bN_S$ in \eqref{colun observation}, the following lemma is a preliminary which gives the distribution of entries in the product of matrices.

\begin{lemma} \label{lemma2}
Given a semi-unitary matrix $\bA\in \C^{d \times N}$ with $\bA \bA^H  = \bI_{d}$, and a random matrix $\bX \in \C^{N \times m}$ with i.i.d. entries of $\cC\cN(0,\sigma^2)$, the product $\bY = \bA \bX \in \C^{d \times m}$ also has i.i.d. entries with distribution of $\cC\cN(0,\sigma^2)$.
\begin{proof}
See Appendix \ref{lemma2_proof}. \hfill \qed
\end{proof}
\end{lemma}

Therefore, considering the noise part in \eqref{colun observation}, i.e., $\bN_S = \bM\widetilde{\bN}$, where $\bM$ is unitary and $\widetilde{\bN}$ has i.i.d. $\cC \cN(0, \sigma^2)$ entries, the conclusion of Lemma \ref{lemma2} can be applied, which verifies that the entries of $\bN_S$ in \eqref{colun observation} are i.i.d. as $\cC \cN(0, \sigma^2)$.

Given the expression in \eqref{colun observation}, the column subspace estimation problem is formulated as,
\vspace{-2.5pt}
\beq
\widehat{\bU}= \argmax \limits _{\bU \in \C^{N_r \times L}}  \lA \bU^H \bY_S   \rA_F^2 ~\text{subject to} ~ \bU^H \bU  = \bI_L, \label{optimal column subpace 1}
\eeq
where one of the optimal solutions of \eqref{optimal column subpace 1} can be obtained by taking the dominant $L$ left singular vectors of $\bY_S$.
Here, the number of paths, $L$, is assumed to be known as a priori. In practice,
it is possible to estimate $L$ by comparing the singular values of $\bY_S$ \cite{eigenValueEst}. Because $\bY_S = \bH_S+\bN_S$ and $\rank(\bH_S)=L$, there will be $L$ singular values of $\bY_S$ whose magnitudes clearly dominate the other singular values. Alternatively, we can set it to $L_\text{sup}$, which is an upper bound on the number of dominant paths  such that $L\leq L_\text{sup}$.\footnote{Due to the limited RF chains, the dimension of channel subspaces for data transmission is less than $\min\{M_{RF}, N_{RF}\}$. Thus, if the path number estimate is larger than $\min\{M_{RF}, N_{RF}\}$, we let it be
$\min\{M_{RF}, N_{RF}\}$.}

Now, we design the receive combiner $\widehat{\bW}$ in \eqref{receiver form} for data transmission to approximate the estimated $\widehat{\bU} \in \C^{N_r \times L}$ in \eqref{optimal column subpace 1}. Specifically, we design the analog combiner $\widehat{\bW}_A \in \C^{N_r \times M_{RF}}$ and digital combiner $\widehat{\bW}_D \in \C^{M_{RF} \times L}$ at the receiver
by solving the following problem
\vspace{-2.5pt}
\beq
\left(\widehat{\bW}_A, \widehat{\bW}_D\right) = \argmin_{\bW_A,\bW_D} \| \widehat{\bU}-\bW_A\bW_D \|_F, \nonumber \\
\text{subject to~~} \la[\bW_A]_{i,j}\ra=\frac{1}{\sqrt{N_r}} .\label{receiver sounder}
\eeq
The problem above can be solved by using the OMP algorithm \cite{spatially} or alternating minimization method \cite{Alternating_Min}. The designed receive combiner is given by $\widehat{\bW} = \widehat{\bW}_A \widehat{\bW}_D \in \C^{N_r \times L}$ with $\widehat{\bW}^H\widehat{\bW}=\bI_L$. The methods in  \cite{spatially, Alternating_Min} have shown to guarantee  the near optimal performance, such as $\mathop{\mathrm{col}}(\widehat{\bW}) \approx \mathop{\mathrm{col}}(\widehat{\bU})$.
The details of our column subspace estimation algorithm are summarized in Algorithm \ref{alg_column}.

In general, $\mathop{\mathrm{col}}(\widehat{\bW})$ is not equal to the column subspace of $\bH$, i.e., $\mathop{\mathrm{col}}( \bU)$ with $\bU \in \C^{N_r \times L}$, due to the noise $\bN_S$ in \eqref{colun observation}.
To analyze the column subspace accuracy $\eta_c(\widehat{\bW})$ defined in \eqref{def eta c}, we introduce the theorem \cite{cai2018} below.
\begin{theorem}[\cite{cai2018}]
 \label{theorem1}
  Suppose $\bX \in \C^{M \times N} (M \ge N)$ is of rank-$r$, and $\widehat{\bX} = \bX + \bN$, where $[\bN]_{i,j}$ is i.i.d. with zero mean and unit variance (not necessarily Gaussian). Let the compact SVD of $\bX$ be
  \beq
  \bX = \bU \bSig \bV^H, \nonumber
  \eeq
  where $\bU \in \C^{M \times r}$, $\bV \in \C^{N \times r}$, and $\bSig \in \C^{r \times r}$. We assume the singular values in $\bSig$ are in descending order, i,e, $\sig_1(\bX) \geq \cdots \geq \sig_r(\bX)$.
  Similarly, we partition the SVD of $\widehat{\bX}$ as
  \beq
  \widehat{\bX} =
  \begin{bmatrix}
    \widehat{\bU} &\widehat{\bU}_{\perp}
  \end{bmatrix}
  \begin{bmatrix}
    \widehat{\bSig}_1 & \mathbf{0} \\
    \mathbf{0} & \widehat{\bSig}_2
  \end{bmatrix}
    \begin{bmatrix}
      \widehat{\bV}^H \\
      \widehat{\bV}_{\perp}^H
    \end{bmatrix}, \nonumber
  \eeq
  where $\widehat{\bU} \in \C^{M \times r}$, $\widehat{\bU}_{\perp} \in \C^{M \times (M-r)}$,
  $\widehat{\bV} \in \C^{N \times r}$, $\widehat{\bV}_{\perp} \in \C^{N \times (N-r)}$, $\widehat{\bSig}_1 \in \C^{r \times r}$, and $\widehat{\bSig}_2 \in \C^{(M-r) \times (N-r)}$.
Then, there exists a constant $C > 0$ such that
\beq
\E\left[ \sigma_r^2( \bU^H \widehat{\bU} ) \right] \ge \left( 1- \frac{CM (\sigma_r^2(\bX)+N)}{\sigma_r^4(\bX)} \right)_{+},\nonumber \\
\E\left[ \sigma_r^2( \bV^H \widehat{\bV} )\right] \ge \left(1- \frac{CN (\sigma_r^2(\bX)+M)}{\sigma_r^4(\bX)}\right)_{+}, \nonumber
\eeq
where the expectation is taken over the random noise $\bN$.
In particular, when the noise is i.i.d. $\cC\cN(0,1)$, it has $C=2$.
\end{theorem}

\begin{algorithm} [t]
\caption{Column subspace estimation}
\label{alg_column}
\begin{algorithmic}[1]
\STATE Input: channel dimension: $N_r$, $N_t$; number of RF chains at receiver: $M_{RF}$; channel paths: $L$; parameter: $m$.
\STATE Initialization: channel use index $k=1$.
\FOR{$i = 1$ to $m$}
    \STATE Set transmit sounder as $\bff_{(i)} = \be_{i}$.
    \FOR{$j= 1$ to $N_r/M_{RF} $}

\STATE Design receive training sounder as
 $\bW_{(i,j)} =  [\bM]_{:, (j-1)M_{RF}+1:jM_{RF}} \bI_{M_{RF}}$.
\STATE Obtain the received signal $\by_{(i,j)}\!\!\!\!=\!\!\!\!\bW_{(i,j)} ^H \bH \bff_{(i)} \!\!+ \!\!\bW_{(i,j)} ^H \bn_{(i,j)}$.
\STATE Update $k=k+1$.
\ENDFOR
 \STATE$\by_{i} = \left[\by_{(i,1)}^T, \cdots,\by^T_{(i,N_r/M_{RF})} \right]^T$.
\ENDFOR
 \STATE$\bY_S = \bM \left[\by_{1},\cdots, \by_{m}\right]$.
\STATE Column subspace $\widehat{\bU}$ is obtained by the dominant $L$ left singular vectors of $\bY_S$.
\STATE Design $\widehat{\bW}$ based on $\widehat{\bU}$ by solving \eqref{receiver sounder}.
\STATE Output: Column subspace estimation $\widehat{\bW}$.
\end{algorithmic}
\end{algorithm}

We have the following proposition for the accuracy of the column subspace estimation in Algorithm \ref{alg_column}.
\begin{proposition} \label{prop1}
If the Euclidean distance $\| \widehat{\bW}  -  \widehat{\bU}\|_F \le \delta_1$ in \eqref{receiver sounder}, then the accuracy of the estimated column subspace matrix $\widehat{\bW}$ obtained from Algorithm \ref{alg_column} is lower bounded as
\beq
\sqrt{\eta_c(\widehat{\bW})} \ge \sigma_L(\widehat{\bU}^H \bU) -\delta_1, \label{mid eq1}
\eeq
where $\bU \in \C^{N_r \times L}$ is the matrix composed of $L$ dominant left singular vectors of $\bH$.
In particular, if $\delta_1 \rightarrow 0$, we have
\beq
\mathbb{E}\left[ \eta_c (\widehat{\bW}) \right] \!\!\!\!&\ge& \!\!\!\! \sigma_L^2(\widehat{\bU}^H \bU) \nonumber\\
&\ge&\!\!\!\! \left ( 1-\frac{ 2N_r (\sigma^2 \sig_L^2(\bH_S)+m \sigma^4)}{\sig_L^4(\bH_S)} \right)_+, \label{column subspace ac}
\eeq
where the $\sig_L(\bH_S)$ is the $L$th largest singular value of $\bH_S$.
\begin{proof}
See Appendix \ref{appendix2}. \hfill \qed
\end{proof}
\end{proposition}

From \eqref{column subspace ac}, the larger the value of $m$ is, the more accurate the column subspace estimation.
Thus, when more columns are used for the column subspace estimation, the estimated column subspace will be more reliable.
In particular, when the noise level is low such that $\sig_L^2(\bH_S)\!\! \gg \! m\sig^2 $ in \eqref{column subspace ac}, we have
\beq
\mathbb{E}\left[ \eta_c (\widehat{\bW}) \right] \ge \left ( 1-\frac{ 2N_r  \sigma^2 }{\sig_L^2(\bH_S)} \right)_+. \nonumber
\eeq
It means that the column subspace estimation accuracy is linearly related to the value of $\sigma^2\!\!  / \! \sig_L^2(\bH_S)$, i.e., $\cO(\text{SNR})$.
On the other hand, when the noise level is high such that $\sig_L^2(\bH_S) \!  \ll\!    m\sig^2 $, the bound in \eqref{column subspace ac} can be written as
\beq
\mathbb{E}\left[ \eta_c (\widehat{\bW})  \right] \ge\left ( 1-\frac{2 N_r m \sigma^4}{\sig_L^4(\bH_S)} \right)_+. \nonumber
\eeq
At low SNR, the column subspace estimation accuracy is quadratically related to $\sigma^4 /\sig_L^4(\bH_S)$, i.e., $\cO(\text{SNR}^2)$.
\begin{remark}
When the number of paths, $L$, increases, the value of $\sigma_L(\bH_S)$ in \eqref{column subspace ac} will decrease, which can be interpreted as follows. When $m, N_r \! \rightarrow \!  \infty$, the entries in $\bH_S \! \in \!  \C^{N_r \times m}$ can be generally approximated as standard Gaussian r.v.s \cite{probability2010}. Moreover, it has been shown in \cite{wei2017upper,SVDbound} that the $L$th largest singular value of $\sigma_L  (\bH_S   )     \!  \propto  \! \frac{N_r+1-L}{\sqrt{N_r}}$  with high probability.
As a result, the accuracy of column subspace estimation will be decreased as $ L $ increases due to \eqref{column subspace ac} of Proposition \ref{prop1}.
\end{remark}

\vspace{-6pt}
\subsection{Estimate the Row Subspace} \label{section row}
In this subsection, we present how to learn the row subspace by leveraging the estimated column subspace matrix $\widehat{\bW}$.
Because we have already sampled the first $m$ columns of $\bH$ in the first stage, we only need to sample the remaining $N_t-m$ columns to estimate the row subspace as shown in Fig. \ref{algorithm_diagram}.

At the $k$th channel use of the second stage, we observe the $(m+k)$th column of $\bH$, $k=1,\ldots,(N_t-m)$.
To achieve this, we employ the transmit sounder as
\beq
\bff_{(k)}=\be_{m+k} \label{s2 tran sounder}.
\eeq
For the receive sounder, given the estimated column subspace matrix $\widehat{\bW}$ in the first stage, we just let the receive sounder of the second stage be $\widehat{\bW} \in \C^{N_r \times L}$.\footnote{It is worth noting that because the estimated column subspace of the first stage is $\widehat{\bW}  \in \C^{N_r \times L}$, thus the dimension for receive sounder of second stage is ${N_r \times L}$ rather than ${N_r \times M_{RF}}$ in \eqref{training}.} It is worth noting $\widehat{\bW}$ is trivially applicable for hybrid precoding architecture since $\widehat{\bW}$ is obtained from \eqref{receiver sounder}.
Therefore, under the transmit sounder $\bff_{(k)}$ in \eqref{s2 tran sounder} and receive sounder $\widehat{\bW}$ in \eqref{receiver sounder}, the observation $\by_{(k)} \in \C^{L \times 1}$ at the receiver can be given by
\beq
\by_{(k)}&=& \widehat{\bW}^H\bH\bff_{(k)} +\widehat{\bW}^H\bn_{(k)}\nonumber \\
&=& \widehat{\bW}^H[\bH]_{:,m+k}+ \widehat{\bW}^H\bn_{(k)}, \label{second sample}
\eeq
where $\bn_{(k) }\in \C^{N_r \times 1}$ is the noise vector with $\bn_{(k)} \sim \cC \cN(\boldsymbol{0}_{N_r},\sigma^2\bI_{N_r})$. Then, the observations $k=1, \ldots,(N_t-m)$  in  \eqref{second sample} are packed into a matrix {$\widehat{\bQ}_C  \in \C^{L \times (N_t-m)}$} as
\vspace{-2.5pt}
\beq
\widehat{\bQ}_C &=& [\by_{(1)},\by_{(2)}, \cdots, \by_{(N_t-m)}] \nonumber \\
&=&  \widehat{\bW}^H(\bH_C + \bN_C), \label{Qc expression}
\eeq
where $\bH_C = \left[[\bH]_{:,m+1},\ldots, [\bH]_{:,N_t}\right] \in \C^{N_r \times (N_t-m)}$, and $\bN_C = [\bn_{(1)}, \ldots, \bn_{(N_t-m)}] \in \C^{N_r \times (N_t-m)}$.

\begin{algorithm} [t]
\caption{Row subspace estimation}
\label{alg_row}
\begin{algorithmic}[1]
\STATE Input: channel dimension: $N_r$, $N_t$;  channel paths: $L$; estimated column subspace: $\widehat{\bW}$;  observations of first stage: $\bY_S$; parameter: $m$.
\STATE Set the receive training sounder as $\widehat{\bW}$.
\FOR{$k = $ 1 to $(N_t-m)$}
\STATE Set the transmit training sounder as $\bff_{(k)}= \be_{m+k}$.

\STATE Obtain the received signal:
\STATE $\by_{(k)}=\widehat{\bW}^H \bH \bff_{(k)} + \widehat{\bW}^H \bn_{(k)}$.
\ENDFOR
\STATE Stack all the observations and \eqref{coe 1}:
\STATE $\widehat{\bQ}_C = [\by_{(1)},\by_{(2)}, \cdots, \by_{(N_t-m)}] $.
\STATE Calculate $\widehat{\bQ}$:
$\widehat{\bQ}=\left [\widehat{\bW}^H\bY_S,\widehat{\bQ}_C\right]$.
\STATE  Row subspace matrix $\widehat{\bV}$ is obtained by the dominant $L$ right singular vectors of $\widehat{\bQ}$.
\STATE Design $\widehat{\bF}$ based on $\widehat{\bV}$ by solving \eqref{transmit precoder design}.
\STATE Output: row subspace estimation $\widehat{\bF}$.
\end{algorithmic}
\end{algorithm}

In addition, given the receive sounder $\widehat{\bW}$ and observations $\bY_S$ of the first stage in \eqref{colun observation}, we define $\widehat{\bQ}_S \in \C^{L \times m}$ as,
 \beq
 \widehat{\bQ}_S = \widehat{\bW}^H \bY_S = \widehat{\bW}^H (\bH_S + \bN_S). \label{coe 1}
 \eeq
Combining \eqref{coe 1} and \eqref{Qc expression} yields {$\widehat{\bQ}\in \C^{L \times N_t}$} expressed as,
\beq
\widehat{\bQ} &=&\left [\widehat{\bQ}_S,\widehat{\bQ}_C\right]  \nonumber \\
&=&\left [ \widehat{\bW}^H (\bH_S + \bN_S), \widehat{\bW}^H (\bH_C + \bN_C)\right] \nonumber\\
 &=&  \underbrace{\widehat{\bW}^H \bH}_{\triangleq \bar{\bQ}} +\underbrace{\widehat{\bW}^H \bN}_{\triangleq \bar{\bN}}, \label{Q exp}
\eeq
where $\bN=[\bN_S,\bN_C] \! \!\in \C^{N_r \times N_t}$, $\bH=[\bH_S,\bH_C] \in \C^{N_r \times N_t}$, $\bar{\bQ} = \widehat{\bW}^H \bH \in \C^{L \times N_t}$, and $\bar{\bN} = \widehat{\bW}^H \bN \in \C^{N_r \times N_t}$. Meanwhile, since $\widehat{\bW}$ is semi-unitary and the entries in ${\bN}$ are i.i.d. with distribution $\cC \cN(0,\sigma^2)$, according to Lemma \ref{lemma2}, the entries in $\bar{\bN}$ are also i.i.d. with distribution $\cC \cN(0,\sigma^2)$.

Now, given the expression $\widehat{\bQ}$ in \eqref{Q exp},
the row subspace estimation problem is formulated as,
\beq
\widehat{\bV}= \argmax \limits _{\bV \in \C^{N_t \times L}}  \|  \widehat{\bQ} \bV   \|_F^2 ~\text{subject to} ~ \bV^H \bV  = \bI_L,  \nonumber
\eeq
where the estimated row subspace matrix $\widehat{\bV} \in \C^{N_t \times L}$ is obtained as the dominant $L$ right singular vectors of $\widehat{\bQ}$.
Similarly, in order to design the precoder $\widehat{\bF}$ in \eqref{receiver form} for data transmission, we need to approximate the estimated row subspace matrix $\widehat{\bV}$ under the hybrid precoding architecture.
Specifically, we design the analog precoder $\widehat{\bF}_A \in \C^{N_t \times N_{RF}}$and digital precoder $\widehat{\bF}_D \in \C^{N_{RF} \times L}$ by solving the following problem
\vspace{-2.5pt}
\beq
\left(\widehat{\bF}_A, \widehat{\bF}_D\right) = \argmin_{\bF_A,\bF_D} \| \widehat{\bV}-\bF_A\bF_D \|_F, \nonumber \\
\text{subject to~~} \la[\bF_A]_{i,j}\ra=\frac{1}{\sqrt{N_t}} .\label{transmit precoder design}
\eeq
Therefore, the transmit precoder is given by $\widehat{\bF} = \widehat{\bF}_A \widehat{\bF}_D \in \C^{N_t \times L}$ with $\widehat{\bF}^H\widehat{\bF}=\bI_L$. Similarly, the method on solving \eqref{transmit precoder design} in \cite{spatially} can guarantee $\text{col}(\widehat{\bF}) \approx \text{col}(\widehat{\bV})$.
 The details of our row subspace estimation algorithm are shown in Algorithm \ref{alg_row}.
We have the following proposition about the estimated row subspace accuracy for Algorithm \ref{alg_row}.

\begin{proposition} \label{lemma row}
If the Euclidean distance $\| \widehat{\bF}  -  \widehat{\bV}\|_F \le \delta_2$ in \eqref{transmit precoder design}, then the accuracy of the estimated row subspace matrix $\widehat{\bF}$ obtained from Algorithm \ref{alg_row} is lower bounded as
\beq
\sqrt{\eta_r(\widehat{\bF})} \ge \sigma_L(\widehat{\bV}^H \bV) -\delta_2, \label{lemma row mid1}
\eeq
where $\bV \in \C^{N_t \times L}$ is the matrix composed of the $L$ dominant right singular vectors of $\bH$. In particular, if $\delta_2 \rightarrow 0$, we have
\beq
\mathbb{E}\left[ \eta_r (\widehat{\bF}) \right] \!\!\!\!&\ge&\!\!\!\! \sigma_L^2(\widehat{\bV}^H \bV) \nonumber \\
&\ge&\!\!\!\! \left ( 1-\frac{2 N_t(\sigma^2 \sig_L^2(\bar{\bQ})+L \sigma^4)}{\sig_L^4(\bar{\bQ})} \right)_+, \label{coefficient matrix}
\eeq
where $\sig_L(\bar{\bQ})$ is the $L$th largest singular value of $\bar{\bQ}$ in \eqref{Q exp}.
\end{proposition}
\begin{proof}
See Appendix \ref{appendix3}. \hfill \qed
\end{proof}

Similar as the column subspace estimation, the row subspace accuracy linearly increases with the SNR, i.e., $\mathcal{O}(\text{SNR})$ at high SNR, and quadratically increases with SNR, i.e., $\cO(\text{SNR}^2)$, at low SNR.
Also, the accuracy of row subspace estimation decreases with the number of paths, $L$.
As the value of $\sig_L(\bar{\bQ})$ in \eqref{coefficient matrix} grows, we can have a more accurate row subspace estimation. Moreover, considering $\bar{\bQ} =\widehat{\bW}^H  \bH$, it is intuitive that the estimated column subspace matrix $\widehat{\bW}$ will affect
the value of $\sig_L(\bar{\bQ})$, and then affect the accuracy of row subspace estimation. Specifically, when $\mathop{\mathrm{col}}(\widehat{\bW}) = \mathop{\mathrm{col}}({\bU}) $, we will have
$\sig_L(\bar \bQ) = \sig_L(\bH)$, which attains the maximum.
In the following, we further discuss the relationship between $\sig_L(\bar{\bQ})$ and $\sig_L(\bH)$.

With the SVD of $\bH$, i.e., $\bH = \bU\bSig\bV^H$, we have $\bar{\bQ} = \widehat{\bW}^H\bH  =   \widehat{\bW}^H \bU \bSig \bV^H$.
Then, the following relationship is true due to the singular value product inequality,
\beq
\sigma_L(\bar{\bQ} ) &\ge& \sigma_L( \widehat{\bW}^H \bU  ) \sigma_L(\bSig \bV^H )    \nonumber \\
&=&  \sigma_L(\widehat{\bW}^H \bU  ) \sigma_L(\bH )  . \label{column affect}
\eeq
Therefore, $\sigma_L(\bar{\bQ} )$ is lower bounded by the product of the $L$th largest singular values of $\widehat{\bW}^H \bU $ and $\bH $. When the estimation of the column subspace becomes accurate, the $\sigma_L(\widehat{\bW}^H \bU  )$ will approach to one.
As a result, the value of $\sigma_L(\bar{\bQ} ) $ is approximately equal to $\sigma_L(\bH )$, resulting in a further enhanced row subspace estimation.
The inequality in \eqref{column affect} reveals that the column subspace estimation affects the accuracy of the row subspace estimation.

Given the estimated column subspace $\widehat{\bW}$ in Algorithm \ref{alg_column} and row subspace $\widehat{\bF}$ in Algorithm \ref{alg_row}, the following lemma shows the subspace estimation accuracy of the proposed SASE, i.e., $\eta(\widehat{\bW},\widehat{\bF})$ defined in \eqref{subspace metric}.
\begin{lemma} \label{subspace proof}
If we assume $\delta_1 \rightarrow 0 $ and $\delta_2 \rightarrow 0$ in \eqref{receiver sounder} and \eqref{transmit precoder design},
the subspace estimation accuracy defined in \eqref{subspace metric} associated with $\widehat{\bW}$ and $\widehat{\bF}$ is lower bounded as
\beq
\eta(\widehat{\bW},\widehat{\bF})\ge \sigma_L^2(\widehat{\bU}^H \bU)\sigma_L^2(\widehat{\bV}^H \bV).
\eeq
\end{lemma}
\begin{proof}
  Using the definition of $\eta(\widehat{\bW},\widehat{\bV})$ in \eqref{subspace metric}, we have the following expressions,
\beq
\eta(\widehat{\bW},\widehat{\bF})   &=  & {\|  \widehat{\bW}^H \bH \widehat{\bF}\|_F^2}/{\tr(\bH^H \bH)}\nonumber \\
&\overset{(a)}{=} & {\|  \widehat{\bU}^H \bH \widehat{\bV}\|_F^2}/{\tr(\bH^H \bH)}\nonumber \\
  &=  & {\|  \widehat{\bU}^H \bU \bSig \bV^H \widehat{\bV}\|_F^2}/{\tr(\bH^H \bH)}\nonumber \\
  &\overset{(b)}{\ge} &  \sigma_L^2(\widehat{\bU}^H \bU)\sigma_L^2(\widehat{\bV}^H \bV),\nonumber
\eeq
where the equality $(a)$ holds for $\delta_1 \rightarrow 0 $ and $\delta_2 \rightarrow 0$, and the inequality $(b)$ holds based on the singular value product inequality.
\hfill \qed
\end{proof}

 Lemma \ref{subspace proof} tells that the power captured by $\widehat{\bW}$ and $\widehat{\bF}$ is lower bounded by the product of  $\sigma_L^2(\widehat{\bU}^H \bU)$ and $\sigma_L^2(\widehat{\bV}^H \bV)$.
 These two parts denotes the two stages in the proposed SASE, which are column subspace estimation and row subspace estimation, respectively. Ideally, when $\text{col}(\widehat{\bU})=\text{col}({\bU})$ and $\text{col}(\widehat{\bV})=\text{col}({\bV})$, we have $\eta(\widehat{\bW},\widehat{\bF})=1$. Nevertheless, the proposed SASE can still achieve nearly optimal $\eta(\widehat{\bW},\widehat{\bF})$. This is because  $\sigma_L^2(\widehat{\bU}^H \bU)$ and $\sigma_L^2(\widehat{\bV}^H \bV)$ are close to one  according to the bounds provided in \eqref{column subspace ac} and \eqref{coefficient matrix}, respectively.

\subsection{Channel Estimation Based on the Estimated Subspaces} \label{SASE channel result}
In this subsection, we introduce a channel estimation method based on the estimated column subspace $\widehat{\bW} \in \C^{N_r \times L}$ and row subspace $\widehat{\bF} \in \C^{N_t \times L}$.
Let the channel estimate be expressed as
\beq
\widehat{\bH} = \widehat{\bW}\widehat{\bR}\widehat{\bF}^H, \label{rep of est}
\eeq
where $\widehat{\bR} \in \C^{L \times L}$. Now, given $\widehat{\bW}$ and $\widehat{\bF}$, it only needs to obtain $\widehat{\bR}$ in an optimal manner.

Recalling the column subspace estimation in Section \ref{Section_column} and row subspace estimation in Section \ref{section row}, the corresponding received signals are expressed as
\beq
\bY_S  &=&\bH_S + \bN_S \nonumber \\
\widehat{\bQ}_C &=& \widehat{\bW}^H\bH_C + \widehat{\bW}^H \bN_C\nonumber.
\eeq
It is worth noting that the entries in $\bN_S$ and $\widehat{\bW}^H \bN_C$ are both i.i.d with distribution $\cC\cN(0,\!\sigma^2)$.
Based on the expression of $\widehat{\bH}$ in \eqref{rep of est}, the maximum likelihood estimation of $\widehat{\bR}$ in \eqref{rep of est} can be obtained through the following least squares problem,
\beq
    &&\!\!\!\!\!\!\!\!\!\!\!\!\!\!\! \!\!\min \limits_{\bR \in \C^{L \times L}} \| \bY_S -  \widehat{\bH}_S \|_F^2+\| \widehat{\bQ}_C-\widehat{\bW}^H \widehat{\bH}_C\|_F^2 \nonumber \\
&&\!\!\!\!\!\!\!\!\!\!\!\!\!\!\! \!\! \text{subject to} ~ \widehat{\bH}_S \!\!=\!\! [\widehat{\bW}\bR\widehat{\bF}^H]_{:,1:m}, ~ \!\! \widehat{\bH}_C \!\!=\!\! [\widehat{\bW}\bR\widehat{\bF}^H]_{:,m+1:N_t}\!.
\label{R estimation}
\eeq
Before discussing how to solve the problem in \eqref{R estimation}, for convenience, we define
\beq
\br &=&\vec(\bR)\in \C^{L^2 \times 1}, \nonumber \\
\by_S &=&\vec(\bY_S)\in \C^{m N_r \times 1} ,\nonumber \\
\widehat{\bq}_C &=&\vec(\widehat{\bQ}_C)\in \C^{(N_t-m)L \times 1},\nonumber \\
\bA_1 &=&([\widehat{\bF}]_{:,1:m}^H)^T \otimes \widehat{\bW} \in \C^{m N_r \times L^2},\nonumber \\
\bA_2 &=&([\widehat{\bF}]_{:,m+1:N_t}^H)^T \otimes \bI_L \in \C^{(N_t-m)L \times L^2}. \nonumber
\eeq
Using the definitions above, the minimization problem in \eqref{R estimation} can be rewritten as
\beq
\min \limits_{\br \in \C^{L^2 \times 1}} \lA \by_S \! - \!\bA_1 \br \rA_2^2 \!+ \! \lA \widehat{\bq}_C\! - \!\bA_2 \br\rA_2^2. \label{form r}
\eeq
The following lemma provides the solution of problem \eqref{form r}.
\begin{lemma} \label{lemma r}
Given the problem below
\beq
\min \limits_{\br \in \C^{L^2 \times 1}} \lA \by_S -\bA_1 \br \rA_2^2 \!+\! \lA \widehat{\bq}_C\!-\!\bA_2 \br\rA_2^2, \nonumber
\eeq
the optimal solution is given by
\beq
\widehat{\br} = (\bA_1^H \bA_1+\bA_2^H \bA_2)^{-1}(\bA_1^H\by_S+\bA_2^H\widehat{\bq}_C). \label{sol r}
\eeq
\end{lemma}
\begin{proof}
The problem is convex with respect to $\br$. Thus, the optimal solution can be obtained by setting the first order derivative of the objective function to zero as
\beq
\bA_1^H(\bA_1\br -\by_S) + \bA_2^H(\bA_2\br - \widehat{\bq}_C) =\mathbf{0}. \label{r equation}
\eeq
The solution of \eqref{r equation} is exactly the result in \eqref{sol r}, which concludes the proof.\hfill \qed
\end{proof}

It is worth noting that after we have obtained the column and row subspace estimates, i.e., $\widehat{\bW}$ and $\widehat{\bF}$, the channel estimation is simply to compute $\widehat{\br}=\vec(\widehat{\bR})$ in \eqref{sol r}. Since the dimension of $\widehat{\bR}$ is much lower than that of $\bH$, the channel estimation complexity is substantially reduced as shown in Lemma \ref{lemma r}.

\vspace{-2pt}
\section{Discussion of Algorithm}
In this section, we analyze the complexity of the proposed SASE method in terms of the  channel use overhead and computational complexity. Moreover, we discuss the application of the SASE in other channel scenarios.

\vspace{-2pt}
\subsection{Channel Use Overhead}
 \begin{table}
		\centering
\caption{Channel Uses of Algorithms}
		\begin{tabular}{|c|c|}\hline
			Algorithms&Number of Channel Uses\\\hline
			SASE&${m N_r}/{M_{RF}} + (N_t-m)$\\
			MF \cite{AlternatingMin}&$\mathcal{O}(L(N_r+N_t)/{M_{RF}})$\\
			SD \cite{ZhangSD} &$\mathcal{O}(L(N_r+N_t)/{M_{RF}})$\\
			Arnoldi \cite{hadi2015}  &$2qN_r/{M_{RF}}+2qN_t/{N_{RF}}$\\
			OMP \cite{OMPchannel} &  $\mathcal{O}(L\ln (G^2)/{M_{RF}})$ \\
            SBL \cite{SBR_channel} &  $\mathcal{O}(L\ln (G^2)/{M_{RF}})$ \\
			ACE \cite{alk} &$s^2L^3\log_s (N_m/L)/{M_{RF}}$\\ \hline
		\end{tabular}
		\label{table use}
	\end{table}

Considering the channel uses in each stage, the total number of channel uses for the SASE is given by
\vspace{-2pt}
\beq
 K_{\text{SASE}} &=&  {m N_r}/{M_{RF}} + (N_t-m). \label{channel uses}
\eeq
Therefore, the number of channel uses grows linearly with the channel dimension, i.e., $\mathcal{O}(N_t)$.
In particular, when we let $m=L$, the number of channel uses in \eqref{channel uses} will be ${L N_r}/{M_{RF}} + (N_t-L)$. Considering that each channel use contributes to $M_{RF}$ observations in the first stage, and $L$ observations in the second stage, the total number of the observations is $LN_r + L(N_t-L)$, which is equivalent to the degrees of freedom of $\rank$-$L$ matrix $\bH \in \C^{N_r \times N_t}$ \cite{matrixCom}.

The numbers of channel uses of the proposed SASE and other benchmarks \cite{OMPchannel,SBR_channel,alk,ZhangSD,AlternatingMin,hadi2015} are compared in Table \ref{table use}. For the angle estimation methods in \cite{OMPchannel,SBR_channel,alk}, the number of required channel uses for the OMP \cite{OMPchannel} and SBL \cite{SBR_channel} is $K_\text{OMP}=K_\text{SBL} = \mathcal{O}(L\ln (G^2)/{M_{RF}})$, where $G$ is the number of grids with $G \ge \max\{N_r,N_t \}$.
The number of channel uses for adaptive channel estimation (ACE) \cite{alk} is $K_\text{ACE}=s^2L^3\log_s (N_m/L)/{M_{RF}}$, where $2\pi/N_m$ with $N_m \ge \max\{N_r,N_t \}$ is the desired angle resolution for the ACE, and $s$ is the number of beamforming vectors in each stage of the ACE.
For the subspace estimation methods in \cite{ZhangSD,AlternatingMin,hadi2015},  the numbers of required channel uses for subspace decomposition (SD) \cite{ZhangSD} and matrix factorization (MF) \cite{AlternatingMin} are $K_\text{SD}=K_\text{MF}=\mathcal{O}(L(N_r+N_t)/{M_{RF}})$,
while it requires $K_\text{Arnoldi} =2qN_r/{M_{RF}}+2qN_t/{N_{RF}}$ channel uses where $q\ge L$ for Arnoldi approach \cite{hadi2015}.
Because the number of estimated parameters of the angle estimation methods such as OMP, SBL, and ACE, is less than that of the proposed SASE, they require slightly fewer channel uses than SASE. Nevertheless, the proposed SASE consumes fewer channel uses than those of  the existing subspace estimation methods \cite{ZhangSD,AlternatingMin,hadi2015} as shown in Table \ref{table use}.

\vspace{-2pt}
\subsection{Computational Complexity}

For the proposed SASE, the computational complexity of the first stage comes from the SVD of ${\bY}_S$, which is $\mathcal{O}(m^2 N_r )$ [28].
The complexity
of the second stage is dominated by the design of $\widehat{\bW}$ in \eqref{receiver sounder}, which is $\mathcal{O}( L D N_r )$, where $D\ge   N_r$ denotes the cardinality of an over-complete
dictionary.
Hence, the overall complexity of the proposed SASE algorithm is $\mathcal{O}(m^2 N_r + L D N_r ) = \mathcal{O}(L D N_r )$.
The computational complexities of benchmarks, i.e., the angle estimation methods OMP \cite{OMPchannel}, SBL \cite{SBR_channel}, and ACE \cite{alk} along with the subspace estimation methods Arnoldi \cite{hadi2015}, SD \cite{ZhangSD}, and MF \cite{AlternatingMin} are compared in Table II, where $K$ denotes the number of channel uses.
For a fair comparison, when comparing the computational complexity, we assume the number of channel uses, $K$, is equal among the benchmarks.
As we can see from Table \ref{table com}, the proposed SASE has the lowest computational complexity.



 \begin{table}
		\centering
\caption{Computational Complexity of Algorithms}
		\begin{tabular}{|c|c|}\hline
			Algorithms&Computational Complexity\\\hline
			SASE&$\mathcal{O}(L D N_r )$\\
			MF \cite{AlternatingMin} &$\mathcal{O}(K M_{RF} L^2(N_r^2+N_t^2))$\\
			SD \cite{ZhangSD}&$\mathcal{O}(K M_{RF} L^2(N_r^2+N_t^2))$\\
			Arnoldi \cite{hadi2015} &$\mathcal{O}(K^2M_{RF}^2/(N_r+N_t))$\\
			OMP \cite{OMPchannel}&$\mathcal{O}(L KM_{RF}G^2)$ \\
            SBL \cite{SBR_channel}&$\mathcal{O}(G^6)$ \\
			ACE \cite{alk}&$\mathcal{O} (K M_{RF}^2 D N_r/(sL) +K N_{RF}^2 D N_t/(sL) )$\\ \hline
		\end{tabular}
		\label{table com}
	\end{table}

\vspace{-2pt}
\subsection{Extension of SASE} \label{extension sec}

In this subsection, we extend the proposed SASE to the 2D mmWave channel model with UPAs.
There are $N_{cl}$ clusters, and
each of cluster is composed of $N_{ray}$ rays. For this model, the mmWave channel matrix is expressed as \cite{BaiChannel,IntroductionM,spatially}
\beq
\bH = \sqrt{\frac{N_r N_t}{N_{cl} N_{ray}}}\sum_{i=1}^{N_{cl}} \sum_{j=1}^{N_{ray}} h_{ij} \ba_r(\phi^r_{ij}, \theta^r_{ij}) \ba_t^H(\phi^t_{ij},\theta^t_{ij}), \label{channel model SC}
\eeq
where $h_{ij}$ represents the complex gain associated with
the $j$th path of the $i$th cluster.
 The $\ba_r(\phi^r_{ij}, \theta^r_{ij}) \in \C^{N_r \times 1}$ and $ \ba_t(\phi^t_{ij},\theta^t_{ij}) \in \C^{N_t \times 1}$ are the receive and transmit array response vectors, where $\phi^r_{ij} (\phi^t_{ij})$ and $ \theta^r_{ij}(\theta^t_{ij})$ denote the azimuth and elevation angles of the receiver (transmitter). Specifically, the $\ba_r(\phi^r_{ij}, \theta^r_{ij}) $ and $ \ba_t(\phi^t_{ij},\theta^t_{ij}) $ are expressed as
 \beq
 \ba_r(\phi^r_{ij}, \theta^r_{ij}) = \frac{1}{\sqrt{N_r}}[1,\cdots, e^{j\frac{2\pi}{\lambda}d(m_r \sin \phi^r_{ij}  \sin \theta^r_{ij}  + n_r \cos \theta^r_{ij}  )} ,\nonumber \\
 \cdots,e^{j\frac{2\pi}{\lambda}d((\sqrt{N_r}-1)\sin \phi^r_{ij}  \sin \theta^r_{ij}  + (\sqrt{N_r}-1) \cos \theta^r_{ij}  )}], \nonumber \\
  \ba_t(\phi^t_{ij}, \theta^t_{ij}) = \frac{1}{\sqrt{N_t}}[1,\cdots, e^{j\frac{2\pi}{\lambda}d(m_t \sin \phi^t_{ij}  \sin \theta^t_{ij}  + n_t \cos \theta^t_{ij}  )}, \nonumber \\
 \cdots,e^{j\frac{2\pi}{\lambda}d((\sqrt{N_t}-1)\sin \phi^t_{ij}  \sin \theta^t_{ij}  + (\sqrt{N_t}-1) \cos \theta^t_{ij}  )}], \nonumber
 \eeq
 where $d$ and $\lambda$ are the antenna spacing and the wavelength, respectively, $0 \le m_r, n_r <\sqrt{N_r}$ and $0 \le m_t, n_t < \sqrt{N_t}$ are the antenna indices in the 2D plane.

For the channel model in \eqref{channel model SC}, it is worth noting that the rank of $\bH$ is at most $N_{cl} N_{ray}$. Using the similar derivations as the proof of Lemma \ref{lemma1}, we can verify that  when $m\ge N_{cl} N_{ray}$, the sub-matrix $\bH_S=[\bH]_{:,1:m}\in \C^{N_r \times m}$ satisfies $\rank(\bH_S) = \rank(\bH)$.
Therefore, it is possible to sample the first $m$ columns of $\bH$ in \eqref{channel model SC}
to obtain column subspace information, and sample the remaining columns to obtain the row subspace information.
This means that the proposed SASE can be extended directly to the channel model given in \eqref{channel model SC}.

In summary, the proposed SASE has no strict limitations to be applied to other channel models if the channel matrix $\bH$ experiences sparse propagation and $\text{col}(\bH_S) = \text{col}(\bH)$. Moreover, because the proposed SASE is an open-loop framework, it can be easily extended to multiuser MIMO downlink scenarios.

\vspace{-4pt}
\section{Simulation Results}

In this section, we evaluate the performance of the proposed SASE algorithm by simulation.
\vspace{-4pt}
\subsection{Simulation Setup}

In the simulation, we consider the numbers of the receive and transmit antennas are $N_r = 36$, and $N_t=144$, respectively, and the numbers of the RF chains at the receiver and transmitter are $M_{RF} =6$ and $N_{RF} =8$, respectively. Without lose of generality, it is assumed that the variance of the complex gain of the $l$th path is $\sigma_{h,l}^2=1,\forall l$. We consider three subspace-based channel estimation methods as the benchmarks, i.e., SD \cite{ZhangSD} and MF \cite{AlternatingMin}, and {Arnoldi \cite{hadi2015}}, where SD and MF aim to recover the low-rank mmWave channel matrix, and Arnoldi is to estimate the dominant singular subspaces of the mmWave channel.
For a fair comparison, the considered benchmarks are to estimate the subspace rather than the parameters such as the angles of the paths.

\vspace{-5pt}
\subsection{Numerical Results}

\begin{figure}
\centering
\includegraphics[width=3.3in, height=2.2in]{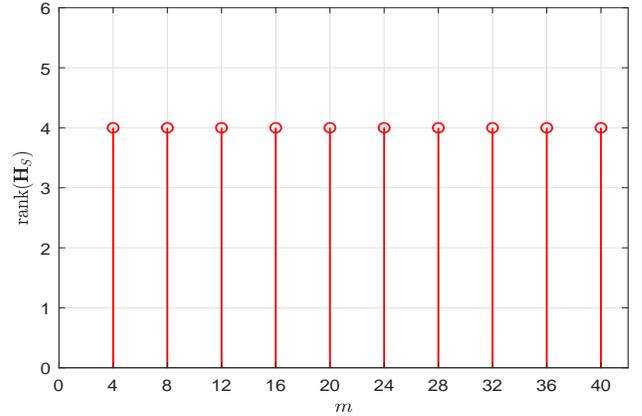}
\caption{$\rank(\bH_S)$ versus $m$ ($N_t = 144; N_r = 36; L = 4$)} \label{figure_rank}
\end{figure}

In order to evaluate the subspace accuracy of different methods, we compute the subspace accuracy $\eta(\widehat{\bW},\widehat{\bF})$ in \eqref{subspace metric}, column subspace accuracy $\eta_c(\widehat{\bW})$ in \eqref{def eta c}, and row subspace accuracy $\eta_r(\widehat{\bF})$ in \eqref{def eta r} for comparison.
We also evaluate the  normalized mean squared error (NMSE) and spectrum efficiency.
The NMSE is defined as
$\text{NMSE}=\E[{\|\bH - \widehat{\bH} \|_F^2}/{\|\bH  \|_F^2}]$,
where $\widehat{\bH}$ denotes the channel estimate. In particular, the channel estimate of the SASE is obtained by the method derived in Section \ref{SASE channel result}.
The spectrum efficiency in \eqref{spectrum efficiency f} is calculated with the combiner $\widehat{\bW}$ and precoder $\widehat{\bF}$, which are designed according to the precoding design techniques provided in \cite{spatially} with the obtained channel estimate $\widehat{\bH}$ via channel estimation.

\begin{figure}[t]
\centering
\small
\setlength{\abovecaptionskip}{0cm}
\begin{tabular}{@{\hskip0pt}c@{\hskip0pt}}
\epsfig{figure=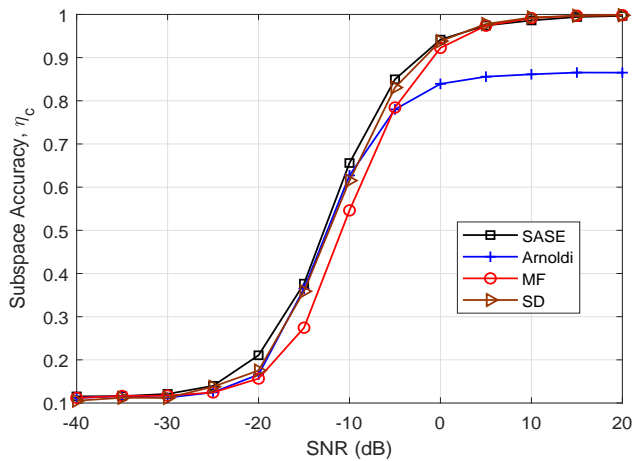,width=3.3in, height=2.4in}\\
(a)\\
\epsfig{figure=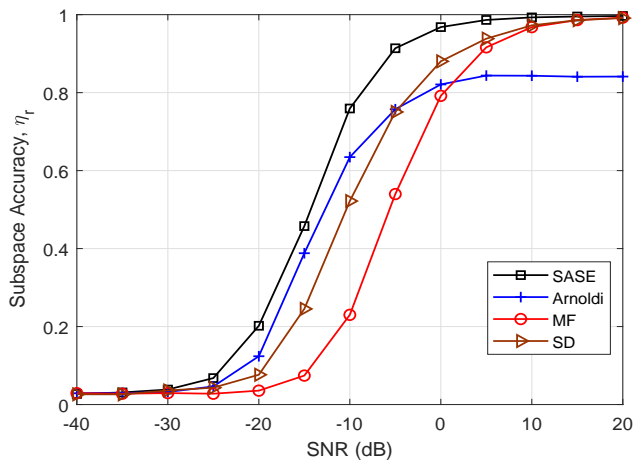,width=3.3in, height=2.4in}\\
 (b)\\
\epsfig{figure=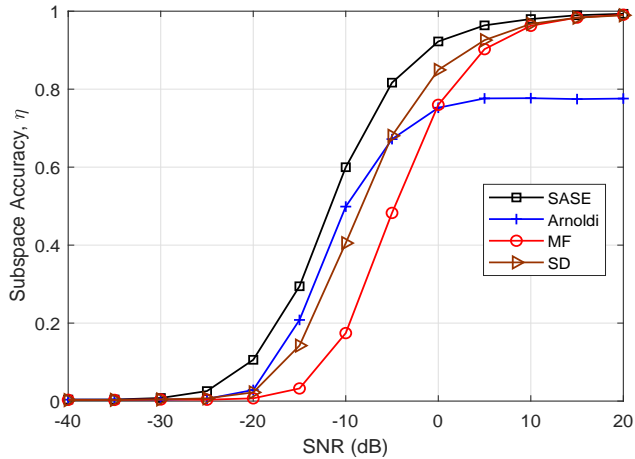,width=3.3in, height=2.4in} \\
(c)\\
\end{tabular}
\caption{The subspace accuracy versus SNR (dB) when $N_t = 144, N_r = 36, L=4, M_{RF} =6, N_{RF} =8,K=244$: (a) Column subspace accuracy $\eta_c$, (b) Row subspace accuracy $\eta_r$, (c) Subspace accuracy $\eta$.
}\label{SNR performance}
\end{figure}

\subsubsection{Equivalence of Subspace}

It is worth noting that the column subspace estimation in Section \ref{Section_column} depends on the fact of subspace equivalence between $\bH_S $ and $\bH$ in \eqref{colun observation}.
We illustrate in Fig. \ref{figure_rank} the rank of $\bH_S$ with different $m$. In this simulation, we set $L=4$ and $m=\{1L,2L,\ldots,10L\}$. It can be seen in Fig. \ref{figure_rank} that the rank of $\bH_S$ is equal to $L$ for all the values of $m$, i.e., the rank of $\bH_S$ is equal to the rank of $\bH$, for $m\ge L$. This validates the fact that $\text{col}(\bH_S)=\text{col}(\bH)$.

\subsubsection{Performance versus Signal-to-Noise Ratio}
In Fig. \ref{SNR performance} and Fig. \ref{SNR performance NMSE},
we compare the performance versus SNR of the proposed SASE algorithm to SD, MF and {Arnoldi} methods.
The number of paths is set as $L=4$.
For a fair comparison, the numbers of channel uses for the benchmarks are
kept approximately equal, i.e., $K=244$.

\begin{figure}[t]
\centering
\small
\setlength{\abovecaptionskip}{0cm}
\begin{tabular}{@{\hskip0pt}c@{\hskip0pt}}
\epsfig{figure=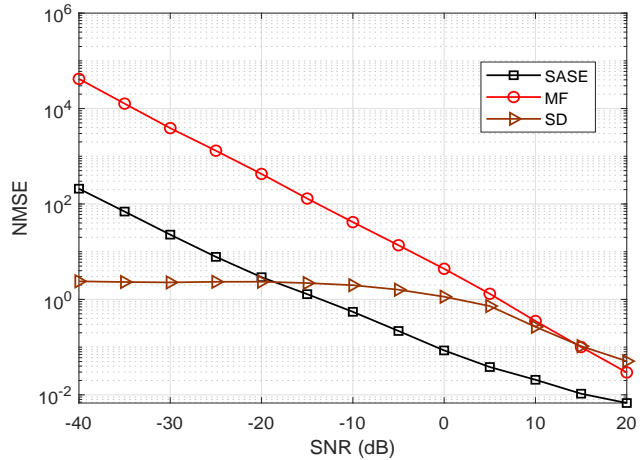,width=3.3in, height=2.4in}\\
(a)\\
\epsfig{figure=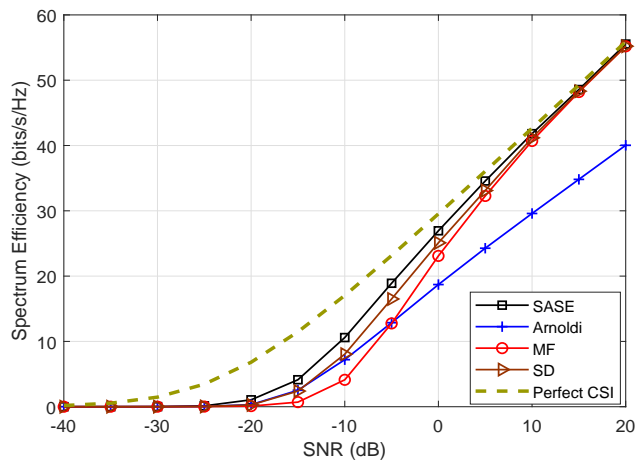,width=3.3in, height=2.4in}\\
 (b)\\
\end{tabular}
\caption{The channel estimation performance versus SNR (dB) when $N_t = 144, N_r = 36, L=4, M_{RF} =6, N_{RF} =8,K=244$: (a) NMSE, (b) Spectrum efficiency.
}\label{SNR performance NMSE}
\end{figure}

In Fig. \ref{SNR performance}(a), the column subspace accuracy $\eta_c$ of the proposed SASE is compared with the benchmarks. As we can see, the SASE and SD methods obtain nearly similar column subspace accuracy, and they outperform over the MF and Arnoldi.
It means that sampling the sub-matrix $\bH_S$ of the channel $\bH$ can provide a robust column subspace estimation.
In Fig. \ref{SNR performance}(b), the row subspace accuracy $\eta_r$ versus SNR is plotted. We found that the proposed SASE outperforms over the others. It verifies that adapting the receiver sounders to the column subspace can surely improve the accuracy of row subspace estimation.
In Fig. \ref{SNR performance}(c), the subspace accuracy $\eta$ defined in \eqref{subspace metric} of the proposed SASE is evaluated. As can be seen that the proposed SASE achieves the most accurate subspace estimation over the other methods. For the SASE, MF and SD, the nearly optimal subspace estimation, i.e., $\eta \approx 1$, can be achieved in the high SNR region ($10 \text{dB}\sim 20 \text{dB}$).
Since the performance of the Arnoldi highly depends on the number of available channel uses, its accuracy is degraded and saturated at high SNR due to the limited channel uses ($K=244$). Thus, the ideal performance of the Arnoldi relies on a large number of channel uses or enough RF chains \cite{hadi2015}.

In Fig. \ref{SNR performance NMSE}(a), the NMSE of the proposed SASE is decreased as the SNR increases. It has similar characteristis as that of the MF, but has much lower value. The NMSE of the SD is almost constant in the low SNR region and decreases in higher SNR region. Overall, the SASE outperforms the SD when $\text{SNR}\ge -15 \text{dB}$. In Fig. \ref{SNR performance NMSE}(b), the spectral efficiency of the SASE is plotted. The curve for perfect CSI with fully digital precoding is plotted for comparison. The proposed SASE achieves the nearly optimal spectrum efficiency among all the methods.
It is observed that the spectrum efficiency of the SASE has a different trend from the NMSE  in Fig. \ref{SNR performance NMSE}(a), while it has similar characteristic as the subspace accuracy in Fig. \ref{SNR performance}(c). The evaluation validates the effectiveness of the SASE in channel estimation to provide good spectrum efficiency.

\begin{figure}[t]
\centering
\small
\setlength{\abovecaptionskip}{0cm}
\begin{tabular}{@{\hskip0pt}c@{\hskip0pt}}
\epsfig{figure=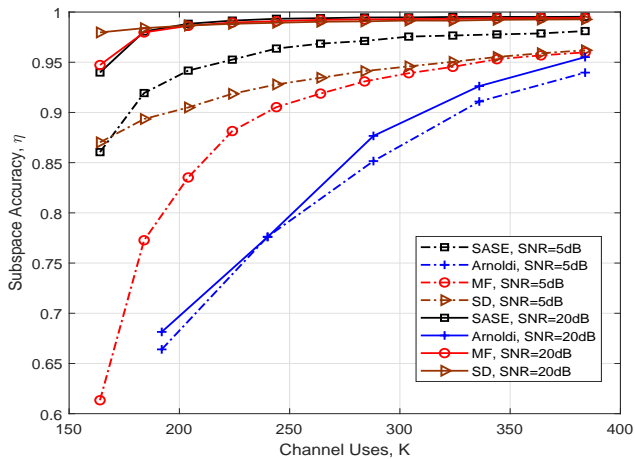,width=3.3in, height=2.4in}\\
(a)  \\
\epsfig{figure=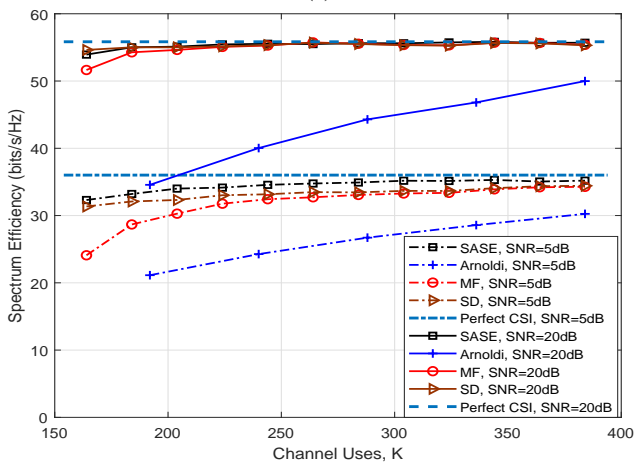,width=3.3in, height=2.4in} \\
(b)  \\
\end{tabular}
\caption{The channel estimation performance versus the number of channel uses $K$ when $N_t = 144, N_r = 36, L=4, M_{RF} =6, N_{RF} =8, \text{SNR}=5\text{dB},20\text{dB}$: (a) Subspace accuracy $\eta$, (b) Spectrum efficiency.
}\label{Uses performance}
\end{figure}

\subsubsection{Performance versus Number of Channel Uses}
In Fig. \ref{Uses performance}, we show the channel estimation performance of the SASE for different numbers of channel uses.  The simulation setting is $L=4, \text{SNR}=5, 20\text{dB}$. The value of $m$ in \eqref{channel uses} is in the set of $\{4,8,\cdots,48\}$. Accordingly, the set of the numbers of channel uses is $K=\{164,184,\cdots,384\}$.

\begin{figure}[t]
\centering
\small
\setlength{\abovecaptionskip}{0cm}
\begin{tabular}{@{\hskip0pt}c@{\hskip0pt}}
\epsfig{figure=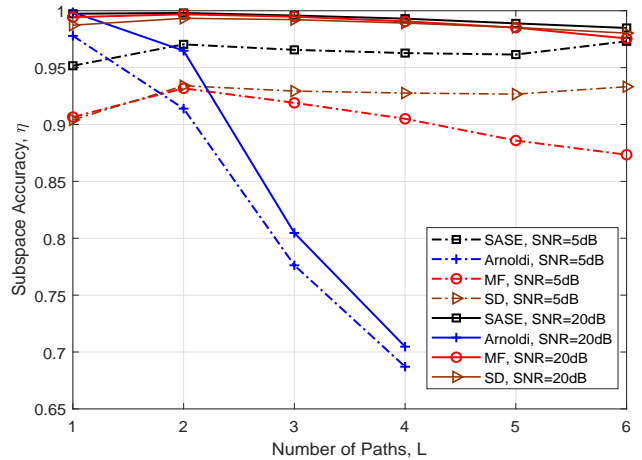,width=3.3in, height=2.4in}\\
(a)  \\
\epsfig{figure=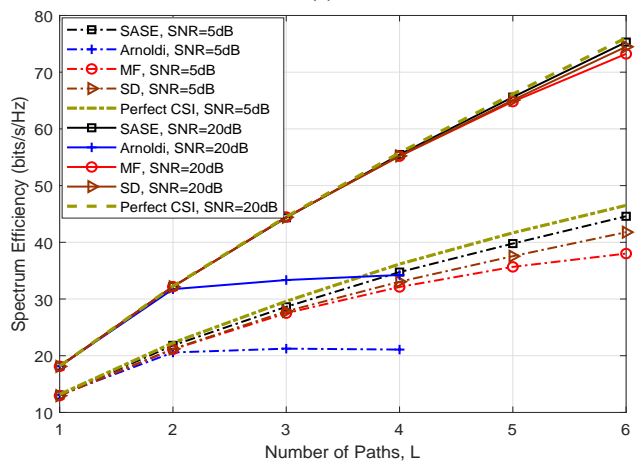,width=3.3in, height=2.4in} \\
(b)  \\
\end{tabular}
\caption{The channel estimation performance versus the number of paths $L$ when $N_t = 144, N_r = 36, M_{RF} =6, N_{RF} =8, K=244 ,\text{SNR}=5\text{dB},20\text{dB}$: (a) Subspace accuracy $\eta$, (b) Spectrum efficiency.
}\label{Paths performance}
\end{figure}

Fig. \ref{Uses performance}(a) shows the subspace estimation performance versus the number of channel uses.
As the number of channel uses increases, the subspace accuracy of all the methods is increased monotonically.
It is worth noting that when $K=164$ ($m=4$), the subspace accuracy of the SASE is slightly lower than that of the SD.
This is because  there are only $m=L=4$ columns sampled for column subspace estimation that affects the column subspace accuracy of the SASE slightly.
Nevertheless, when $m\ge 8$, i.e., $K\ge 184$, the SASE obtains the most accurate subspace estimation, i.e., $\eta\approx 1$, among all the methods.
In particular, when the SNR is moderate, i.e., SNR$=5$dB, the SASE clearly outperforms over the other methods. This means that the SASE requires less channel uses to provide a robust subspace estimation.

Fig. \ref{Uses performance}(b) shows the spectrum efficiency versus the number of channel uses. The curve for perfect CSI with fully digital precoding is also plotted for comparison. Again, the SASE achieves nearly optimal spectrum efficiency compared to the other methods. The performance gap between the SASE and the other methods are more noticeable at SNR$=5$dB. In particular, as seen in the figure, when the number of channel uses, $K\ge 244$, the performance gap between the SASE and perfect curve at SNR$=5$dB is less than $1.5$bits/s/Hz.

\subsubsection{Performance versus Number of Paths}
In Fig. \ref{Paths performance}, we evaluate the estimation performance of the SASE for different numbers of paths, $L$. The number of channel uses is $K=244$ and $\text{SNR}=5,20\text{dB}$. Due to the limited number of channel uses, the Anorldi method can not perform the channel estimation for $L\ge 5$. Thus, we only show the performance of the Arnoldi for $L\le 4$.

In Fig. \ref{Paths performance}(a), the subspace accuracy $\eta$ of different methods versus number of paths, $L$, is illustrated. As we can see, the SASE, SD and MF achieve a more accurate subspace estimation compared to the Arnoldi.
It is seen that the Arnoldi has a sharp decrease in the accuracy for $L>2$. It means that the Arnoldi can provide a good channel estimate only for $L\le 2$ with the use of $K=244$ channel uses.
When SNR$=5$dB, the SASE outperforms over the other methods.
When the SNR is high, i.e., SNR$=20$dB, for the proposed SASE, the subspace accuracy decreases slightly with the number of paths, $L$, which verifies our discussion about the effect of $L$ in Remark 2 of Section III.

\begin{figure}[t]
\centering
\small
\setlength{\abovecaptionskip}{0cm}
\begin{tabular}{@{\hskip0pt}c@{\hskip0pt}}
\epsfig{figure=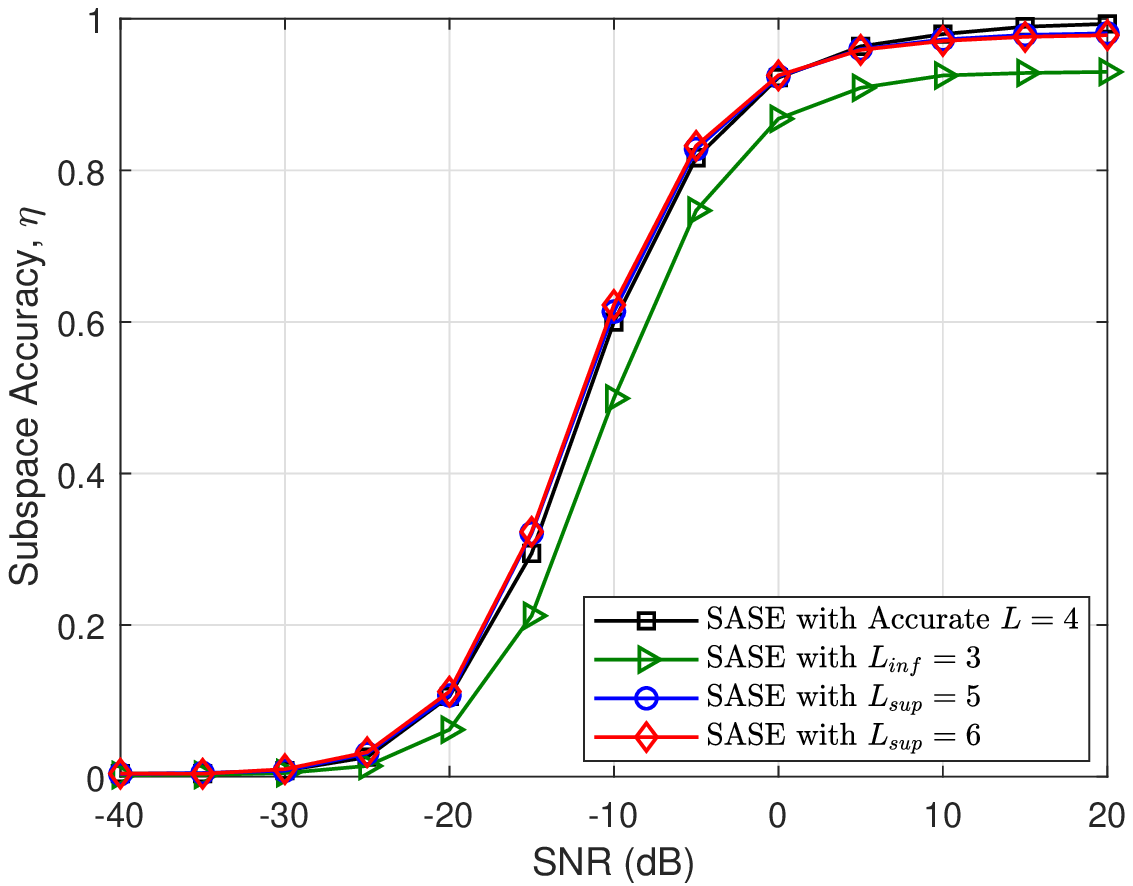,width=3.3in, height=2.4in}\\
(a)\\
\epsfig{figure=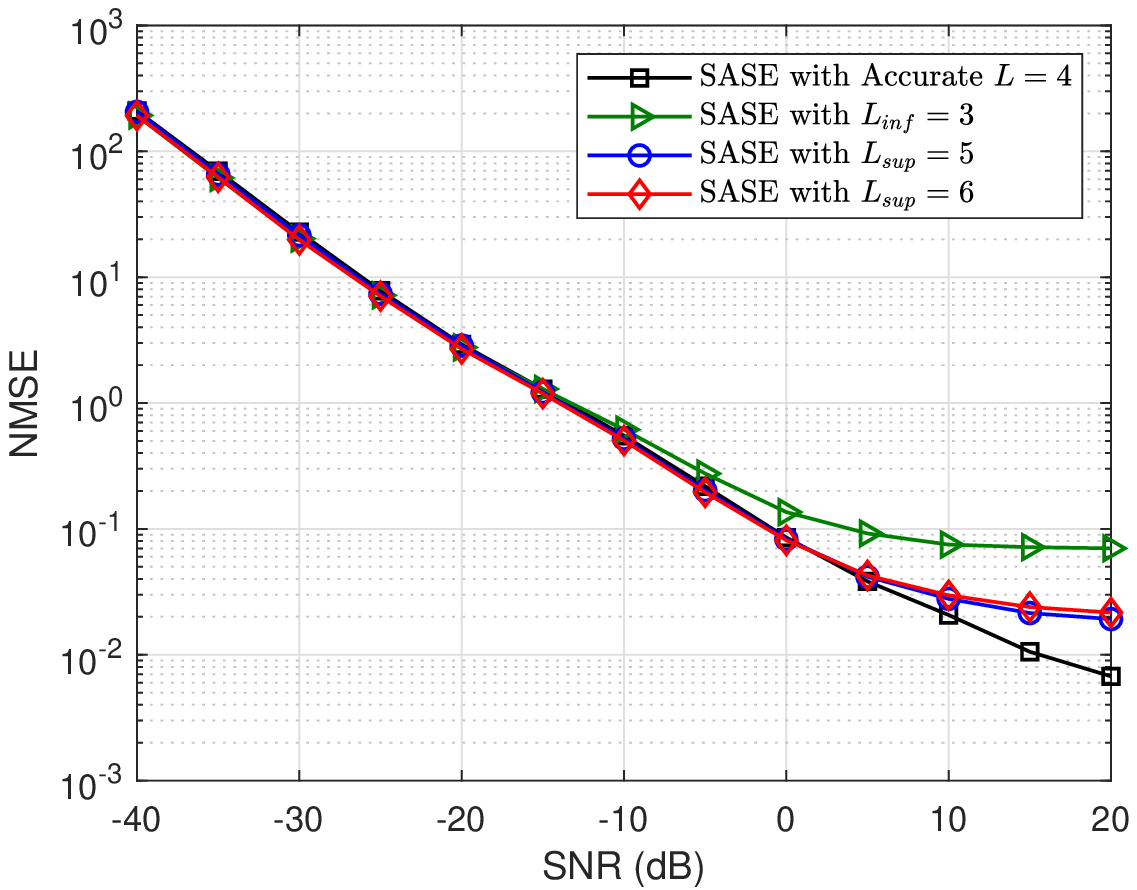,width=3.3in, height=2.4in}\\
 (b)\\
\epsfig{figure=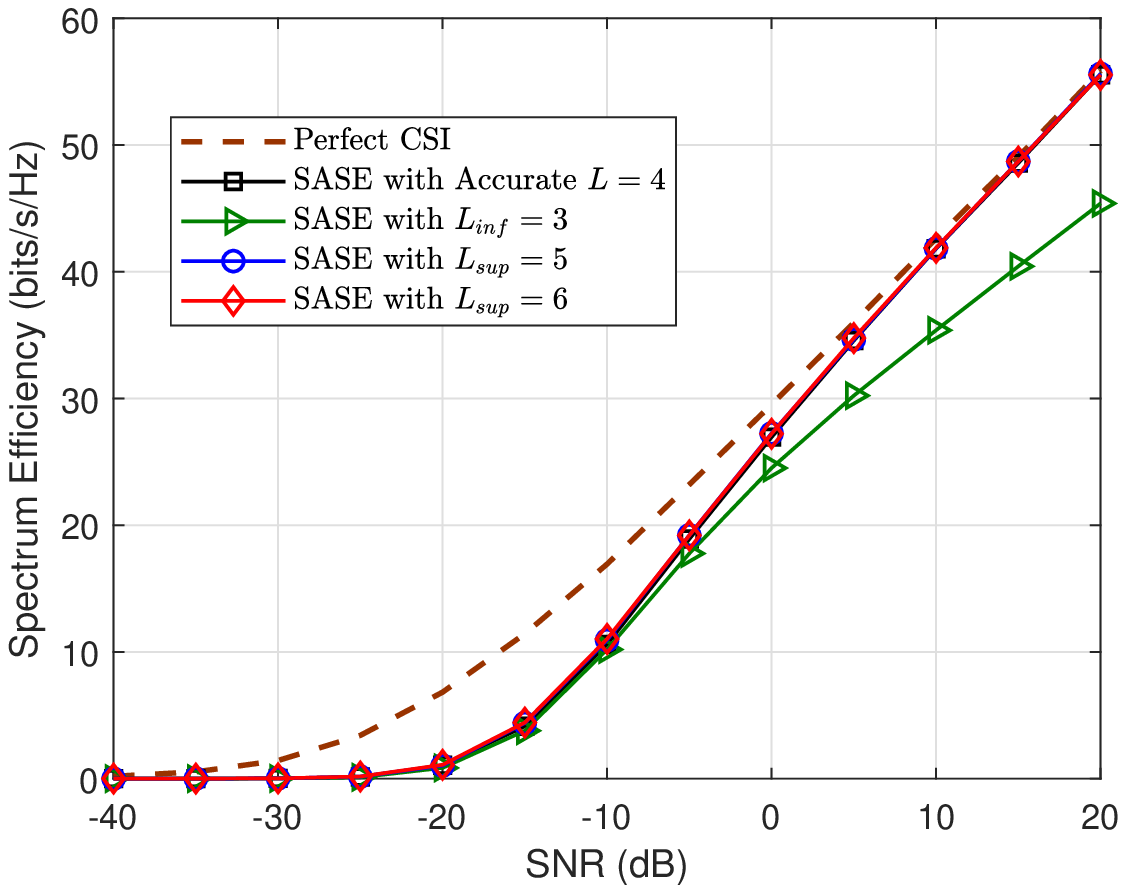,width=3.3in, height=2.4in} \\
(c)\\
\end{tabular}
\caption{The channel estimation performance with inaccurate path information when $N_t = 144; N_r = 36; L = 4; M_{RF} \!=\!6; N_{RF} =8;K=244$: (a) Subspace accuracy $\eta$, (b) NMSE, (c) Spectrum efficiency.
}\label{different_LD}
\end{figure}

In Fig. \ref{Paths performance}(b), the spectrum efficiency versus number of paths, $L$, is shown. Apart from the Arnoldi, the spectrum efficiency achieved by the SASE, MF and SD increases with the number of paths.  When the SNR is high, i.e., SNR$=20$dB, the SASE, MF and SD can achieve nearly optimal performance. When the SNR is moderate, i.e., SNR$=5$dB, the proposed SASE achieves the highest spectrum efficiency among all the methods.
Moreover, for the SASE, MF and SD, their performance gaps with the curve of perfect CSI is getting wider as $L$ increases. Nevertheless, the spectrum efficiency of the SASE is more closer than the other methods, which implies that the SASE can leverage the property of limited number of paths in mmWave channels more effectively than the other methods.

\subsubsection{Performance versus Inaccurate Path information}
Thus far, in the previous simulations, we have assumed the number of paths, $L$, is known a priori. In Fig. \ref{different_LD}, we evaluate the performance of the SASE under the situation that the accurate path information is not available. As discussed in Section III-A, we utilize the upper bound of the number of paths for simplicity, where we let $L_\text{sup}=\{5, 6 \}$ while $L=4$.\footnote{If the $L_\text{sup}$  with $L_\text{sup}\ge L$ is utilized for SASE, the estimated subspaces will be $\widehat{\bW}\in \C^{N_r \times L_\text{sup}}$ and $\widehat{\bF}\in \C^{N_t \times L_\text{sup}}$. For a fair comparison, we choose the dominant $L$ modes in $\widehat{\bW}$ and $\widehat{\bF}$ when evaluating the performance.} For a clear illustration, we also evaluate the performance of proposed SASE by using the lower bound of number of paths, i.e., $L_\text{inf}=3$.
As can be seen in Fig. \ref{different_LD}, compared to the case of $L_\text{inf}=3$,  using the upper bound $L_\text{sup}=\{5,6\}$ for SASE achieves a similar performance as the accurate path information of $L=4$.
In particular, it is noted in Fig. \ref{different_LD}(a) and Fig. \ref{different_LD}(b)  that the estimation performance of $L_\text{sup}$ is slightly worse than that of accurate path information when SNR is high, while it is marginally better when SNR is low.
This is because using inaccurate path information $L_\text{sup}$ with $L_\text{sup} \ge L$ does not affect the column subspace estimation, but according to Proposition \ref{lemma row}, it provides worse row subspace estimation at high SNR  and more accurate row subspace estimation at low SNR.\footnote{It is worth noting that if $L_\text{sup}$ is utilized for SASE, according to Proposition \ref{lemma row}, the row subspace accuracy is bounded as $\mathbb{E}[ \eta_r (\widehat{\bF}) ]  \ge \left ( 1-{2 N_t(\sigma^2 \sig_L^2(\bar{\bQ})+L_\text{sup} \sigma^4)}/{\sig_L^4(\bar{\bQ})} \right)_+$. The statements can be verified easily through analyzing this row subspace accuracy bound.
 }
Nevertheless, in overall, the performance of proposed SASE is not sensitive to the inaccurate path number.

\section{Conclusion}

In this paper, we formulate the mmWave channel estimation as a subspace estimation problem and propose the SASE algorithm. In the SASE algorithm, the channel estimation task is divided into two stages: the first stage is to obtain the column channel subspace, and in the second stage, based on the acquired column subspace, the row subspace is estimated with optimized training signals.
By estimating the column and row subspaces sequentially, the computational complexity of the proposed SASE was reduced substantially to $\mathcal{O}(LDN_r)$ with $D\ge N_r$.
It was analyzed that $\cO(N_t)$ channel uses are sufficient to guarantee subspace estimation accuracy of the proposed SASE.
By simulation, the proposed SASE has better subspace accuracy, lower NMSE, and higher spectrum efficiency than those of the existing subspace methods for practical SNRs.

\appendices
\section{Proof of Lemma \ref{lemma1}} \label{appendix1}
From the mmWave channel model in \eqref{matrix expression}, when
the angles $\{\theta_{t,l} \}_{l=1}^L$ and $\{\theta_{r,l} \}_{l=1}^L$ are distinct,
\beq
\rank(\bA_t) = \rank(\bA_r)=L, \nonumber
\eeq
which holds due to the fact that $\bA_t$ and $\bA_r$ are both Vandermonde matrices. Then, $\bH_S = \bH\bS$ can be expressed as
\beq
\bH_S = \bA_r \diag(\bh) \bA_t^H \bS. \nonumber
\eeq
Combining the rank inequality of matrix product $\rank(\bH_S) \le \rank(\bH) = L$ and the following lower bound,
\beq
 \rank(\bH_S) &\ge& \rank(\bA_r \diag(\bh)) +\rank( \bA_t^H \bS) -L \nonumber \\
 &= &\rank( \bA_t^H \bS), \nonumber
\eeq
yields $L \geq \rank(\bH_S) \geq \rank(\bA_t^H \bS)$.
Therefore, in order to show $\mathop{\mathrm{col}}(\bH_S) = \mathop{\mathrm{col}}(\bH)$, namely, $\rank(\bH_S) = L$, it suffices to show that $\rank(\bA_t^H\bS)=L$.
Considering that $\bA_t^H\bS$ is a Vandermonde matrix, it has $\rank(\bA_t^H\bS)=L$.
This completes the proof. \hfill \qed

\section{Proof of Lemma \ref{lemma2}} \label{lemma2_proof}
It is trivial that the entries in $\bY$ follow the identical distribution of $\cC\cN(0,\sigma^2)$. Therefore, it remains to show that all the entries in $\bY$ are independent. Because of the typical property of Gaussian distribution, it suffices to prove that they are uncorrelated. For any $i \neq j $ or $m \neq n $, the following holds,
  \beq
  \E\left[  [\bY]_{i,m}  [\bY]_{j,n} \right] &= &\E\left[ \bA_{i,:} [\bX]_{:,m} [\bX]_{:,n}^H [\bA]_{j,:}^H \right] \nonumber\\
  &= & 0. \nonumber
  \eeq
  Therefore, the entries in $\bY$ are uncorrelated and thus independent, which concludes the proof. \hfill \qed

\section{Proof of Proposition \ref{prop1}} \label{appendix2}
Based on the definition of $\eta_c(\widehat{\bW})$ in \eqref{def eta c}, it has
\beq
\sqrt{\eta_c(\widehat{\bW})} &=& \sqrt{\frac{\tr( \widehat{\bW}^H \bH \bH^H \widehat{\bW})}{\tr(\bH^H \bH)}} \nonumber\\
&=& \frac{\| (\widehat{\bW} -\widehat{\bU} + \widehat{\bU})^H\bH \|_F}{\| \bH \|_F} \nonumber\\
 &\overset{(a)}{\ge} &\frac{\| \widehat{\bU} ^H\bH \|_F}{\| \bH \|_F} - \frac{\| (\widehat{\bW} -\widehat{\bU})^H\bH \|_F}{\| \bH \|_F}  \nonumber \\
 &= &\frac{\| \widehat{\bU} ^H\bU \bSig \bV^H \|_F}{\| \bH \|_F} - \frac{\| (\widehat{\bW} -\widehat{\bU})^H\bH \|_F}{\| \bH \|_F}  \nonumber \\
 &\overset{(b)}{\ge} & \sigma_L(\widehat{\bU}^H \bU) - \|\widehat{\bW} -\widehat{\bU}\|_2,\nonumber\\
 &{\ge} & \sigma_L(\widehat{\bU}^H \bU) -\delta_1, \label{mid eq1 1}
 \eeq
where the inequality $(a)$ holds from the triangle inequality, and the inequality $(b)$ comes from the fact that for $\bA\in \C^{n \times n}$ with $\rank(\bA)=n$ and $\bB\in \C^{n \times k}$, $\| \bA \bB\|_F^2 \geq \sig_n^2(\bA) \| \bB \|_F^2$, where the latter follows by
   $\| \bA \bB\|_F^2  = \sum_{i=1}^{k} \| \bA [\bB]_{:,i} \|_2^2 \ge  \sum_{i=1}^{k} \sigma_{n}^2(\bA) \| [\bB]_{:,i} \|_2^2 =\sigma_{n}^2(\bA)  \| \bB \|_F^2$. Thus, this concludes the proof for the inequality in \eqref{mid eq1}.

   Then, by letting $\delta_1 \rightarrow 0$ in \eqref{mid eq1}, we take expectation of squares of both sides in \eqref{mid eq1 1}, then it has the following
 \beq
\mathbb{E}\left[ \eta_c (\widehat{\bW}) \right] &\ge & \mathbb{E}\left[  \sigma_L^2(\widehat{\bU}^H \bU) \right]\nonumber\\
&\overset{(c)}{\ge} & \!\!\left ( 1-\frac{ 2N_r (\sigma^2 \sig_L^2(\bH_S)+m \sigma^4)}{\sig_L^4(\bH_S)} \right)_+ \!\!, \label{final col ac}
\eeq
where the inequality $(c)$ holds from Theorem \ref{theorem1}, and this concludes the proof.
\hfill \qed

\section{Proof of Proposition \ref{lemma row}} \label{appendix3}
Recall that the row subspace matrix $\widehat{\bV}$ is given by the right singular matrix of $\widehat{\bQ}=\bar{\bQ} + \bar{\bN}$ in \eqref{Q exp}, and  the elements in  $ \bar{ \bN}$ are i.i.d. with each entry being $\cC \cN(0,\sigma^2)$ according to Lemma \ref{lemma2}. Thus, Theorem \ref{theorem1} is applied, which gives
\beq
\mathbb{E}\left[ \sigma_L^2 (\widehat{\bV}^H \bV) \right] \ge \left ( 1-\frac{2 N_t(\sigma^2 \sig_L^2(\bar{\bQ})+L \sigma^4)}{\sig_L^4(\bar{\bQ})} \right)_+. \label{row 2}
\eeq
Then, based on the subspace accuracy metric in \eqref{def eta r}, it has
\beq
\sqrt{\eta_r(\widehat{\bF})} &=& \sqrt{\frac{\tr( \widehat{\bF}^H \bH^H \bH \widehat{\bF})}{\tr(\bH^H \bH)}} \nonumber\\
&=& \frac{\|  \bH (\widehat{\bF} -\widehat{\bV} + \widehat{\bV}) \|_F}{\| \bH \|_F} \nonumber\\
 &{\ge} &\frac{\| \bH \widehat{\bV} \|_F}{\| \bH \|_F} - \frac{\| \bH (\widehat{\bF} -\widehat{\bV})\|_F}{\| \bH \|_F}  \nonumber \\
 &= &\frac{\| \bU \bSig \bV^H  \widehat{\bV} \|_F}{\| \bH \|_F} - \frac{\| \bH  (\widehat{\bF} -\widehat{\bV})\|_F}{\| \bH \|_F}  \nonumber \\
 &{\ge} & \sigma_L(\widehat{\bV}^H \bV) - \|\widehat{\bF} -\widehat{\bV}\|_2,\nonumber\\
 &{\ge} & \sigma_L(\widehat{\bV}^H \bV) -\delta_2. \label{lemma row mid1 1}
 \eeq
Thus, the inequality \eqref{lemma row mid1} is proved. Moreover, under the condition $\delta_2 \! \rightarrow \! 0$, taking expectation of the squares of both sides in \eqref{lemma row mid1 1} yields
\beq
\mathbb{E}\left[ \eta_r (\widehat{\bF}) \right]  &\ge& \mathbb{E}\left[ \sigma_L^2 (\widehat{\bV}^H \bV) \right] \nonumber \\
&\overset{(a)}{\ge}& \left ( 1-\frac{2 N_t(\sigma^2 \sig_L^2(\bar{\bQ})+L \sigma^4)}{\sig_L^4(\bar{\bQ})} \right)_+, \nonumber
\eeq
where the inequality $(a)$ holds from \eqref{row 2}. This concludes the proof for the row estimation accuracy bound in \eqref{coefficient matrix}. \hfill \qed

\bibliographystyle{IEEEtran}

\bibliography{IEEEabrv,reference}

\clearpage

\end{spacing}

\end{document}